\documentclass[12pt,draftclsnofoot,onecolumn]{IEEEtran}
\makeatletter
\newcommand\semihuge{\@setfontsize\semihuge{22.3}{22}}
\makeatother

\makeatletter
\def\underbracex#1#2{\mathop{\vtop{\m@th\ialign{##\crcr
				$\hfil\displaystyle{#2}\hfil$\crcr
				\noalign{\kern3\p@\nointerlineskip}%
				#1\crcr\noalign{\kern3\p@}}}}\limits}

\def\underbracea{\underbracex\upbracefilla}

\def\upbracefilla{$\m@th \setbox\z@\hbox{$\braceld$}%
	\bracelu\leaders\vrule \@height\ht\z@ \@depth\z@\hfill 
	\kern\p@\vrule \@width\p@\kern\p@\vrule \@width\p@\kern\p@\vrule \@width\p@
	$}

\def\upbracefillb{$\m@th \setbox\z@\hbox{$\braceld$}%
	\vrule \@width\p@\kern\p@\vrule \@width\p@\kern\p@\vrule \@width\p@\kern\p@
	\leaders\vrule \@height\ht\z@ \@depth\z@\hfill\bracerd
	\braceld\leaders\vrule \@height\ht\z@ \@depth\z@\hfill
	\kern\p@\vrule \@width\p@\kern\p@\vrule \@width\p@\kern\p@\vrule \@width\p@
	$}

\def\upbracefillc{$\m@th \setbox\z@\hbox{$\braceld$}%
	\vrule \@width\p@\kern\p@\vrule \@width\p@\kern\p@\vrule \@width\p@\kern\p@
	\leaders\vrule \@height\ht\z@ \@depth\z@\hfill
	\kern\p@\vrule \@width\p@\kern\p@\vrule \@width\p@\kern\p@\vrule \@width\p@
	$}

\def\upbracefilld{$\m@th \setbox\z@\hbox{$\braceld$}%
	\vrule \@width\p@\kern\p@\vrule \@width\p@\kern\p@\vrule \@width\p@\kern\p@
	\leaders\vrule \@height\ht\z@ \@depth\z@\hfill\braceru$}

\def\underbracebd{\underbracex\upbracefillbd}
\def\upbracefillbd{$\m@th \setbox\z@\hbox{$\braceld$}%
	\vrule \@width\p@\kern\p@\vrule \@width\p@\kern\p@\vrule \@width\p@\kern\p@
	\bracerd\braceld
	\leaders\vrule \@height\ht\z@ \@depth\z@\hfill\braceru$}

\makeatother

\usepackage{tabularx}

\usepackage{algpseudocode}
\usepackage{algorithm}
\usepackage{algorithmicx}

\usepackage{lipsum} 
\usepackage{arydshln}
\usepackage[dvips]{color}
\usepackage{comment}
\usepackage{todonotes}
\usepackage{epsf}
\usepackage{epsfig}
\usepackage{times}
\usepackage{epsfig}
\usepackage{graphicx}
\usepackage{bbold}
\usepackage{bbm}
\usepackage{mathtools}
\usepackage{mathrsfs}
\usepackage{amssymb}
\usepackage{pdfpages}
\usepackage{epstopdf}
\newfloat{algorithm}{t}{lop}

\usepackage{amsmath,amssymb}
\usepackage[makeroom]{cancel}
\usepackage{tikz}

\usepackage{dsfont}
\usepackage{lettrine} 

\usepackage{epsfig,algorithm,algpseudocode,amsthm,cite,url}
\newtagform{blue}{\color{blue}(}{)}
\usepackage{subcaption}
\allowdisplaybreaks
\usepackage{csquotes}

\usepackage{verbatim}
\usepackage[english]{babel}
\usepackage{diagbox}
\DeclareMathOperator*{\E}{\mathbb{E}}
\DeclareMathOperator*{\argmax}{arg\,max}
\DeclareMathOperator*{\argmin}{arg\,min}

\usepackage[justification=centering]{caption}
\usepackage{verbatim}

\newtheorem{corollary}{Corollary}
\newtheorem{theorem}{\bf Theorem}

\newtheorem{proposition}{\bf Proposition}

\newtheorem{definition}{\bf Definition}
\newtheorem{remark}{Remark}

\begin{document}
	%
	\title{Cyber-Physical Security and Safety of Autonomous Connected Vehicles: Optimal Control Meets Multi-Armed Bandit Learning}
	\IEEEoverridecommandlockouts
	\author{\IEEEauthorblockN{\normalsize Aidin Ferdowsi, \emph{Student Member, IEEE}, Samad Ali, \emph{Student Member, IEEE}, \\
	Walid Saad, \emph{Fellow, IEEE}, and Narayan B. Mandayam, \emph{Fellow, IEEE}\vspace{-20mm}}
		\thanks{This research was supported by the U.S. National Science Foundation under Grants OAC-1541105 and IIS-1633363.}
		\thanks{Aidin Ferdowsi and Walid Saad are with Wireless@VT, Bradley Department of Electrical and Computer Engineering, Virginia Tech, Blacksburg, VA, USA, {\tt\small \{aidin, walids\}@vt.edu}. Samad Ali is with Centre for Wireless Communications (CWC), University of Oulu, Finland, {\tt\small samad.ali@oulu.fi}. Narayan B. Mandayam is with WINLAB, Dept. of ECE, Rutgers University, New Brunswick, NJ, USA, {\tt\small narayan@winlab.rutgers.edu}}
	}
	\maketitle
	
	\maketitle
	

	%
	\IEEEpeerreviewmaketitle
	
\begin{abstract}	
	Autonomous connected vehicles (ACVs) rely on intra-vehicle sensors such as camera and radar as well as inter-vehicle communication to operate effectively. This reliance on cyber components exposes ACVs to cyber and physical attacks in which an adversary can manipulate sensor readings and physically take control of an ACV. In this paper, a comprehensive framework is proposed to thwart cyber and physical attacks on ACV networks. First, an optimal safe controller for ACVs is derived to maximize the street traffic flow while minimizing the risk of accidents by optimizing ACV speed and inter-ACV spacing. It is proven that the proposed controller is robust to physical attacks which aim at making ACV systems instable. To improve the cyber-physical security of ACV systems, next, data injection attack (DIA) detection approaches are proposed to address cyber attacks on sensors and their physical impact on the ACV system. To comprehensively design the DIA detection approaches, ACV sensors are characterized in two subsets based on the availability of a-priori information about their data. For sensors having a prior information, a DIA detection approach is proposed and an optimal threshold level is derived for the difference between the actual and estimated values of sensors data which enables ACV to stay robust against cyber attacks. For sensors having no prior information, a novel multi-armed bandit (MAB) algorithm is proposed to enable ACV to securely control its motion. Collectively, the proposed DIA detection approaches minimize the vulnerability of ACV sensors against cyber attacks while maximizing the ACV system's physical robustness. Simulation results show that the proposed optimal safe controller outperforms current state of the art controllers by maximizing the robustness of ACVs to physical attacks. The results also show that the proposed DIA detection approaches, compared to Kalman filtering, can improve the security of ACV sensors against cyber attacks and ultimately improve the physical robustness of an ACV system.
\end{abstract}
\section{Introduction}
Intelligent transportation systems (ITS) will encompass autonomous connected vehicles (ACVs), roadside smart sensors (RSSs), vehicular communications, and even drones \cite{ferdowsi2017deep,Mozaffari2016,zeng2018joint,challita2018artificial}. To operate autonomously in future ITS, ACVs must process a large volume of data collected via sensors and communication links. Maintaining reliability of this data is crucial for road safety and smooth traffic flow \cite{Amoozadeh2015,Parvez2018,Husak,Ferdowsicolonel}. However, this reliance on communications and data processing renders ACVs highly susceptible to cyber-physical attacks. In particular, an attacker can possibly interject the ACV data processing stage, inject faulty data, and ultimately induce accidents or compromise the road traffic flow\cite{Kargl2008}. As demonstrated in a real-world experiment on a Jeep Cherokee in \cite{Jeep2015Greenberg}, ACVs are largely vulnerable to cyber attacks that can control their critical systems, including braking and acceleration. Naturally, by taking control of ACVs, an adversary can not only impact the compromised ACV, but it can also reduce the flow of other vehicles and cause a non-optimal ITS operation. This, in turn, motivates a holistic study for joint cyber and physical impacts of attacks on ACV systems. 

Recently, a number of security solutions have been proposed for addressing intra-vehicle network and vehicular communication cyber security problems \cite{Woo2015,obd2016Narayanan,Calandriello2011,Kim2010,Sun2017,2017Eltayeb,PETRILLO2018,ferdowsi2018deep,Tuohy2015}. In \cite{Woo2015}, the authors showed that long-range wireless attacks on the current security protocols of ACVs can disrupt their controller area network (CAN). Furthermore, the work in \cite{obd2016Narayanan}, proposed a data analytics approach for the intrusion detection problem by applying a hidden Markov model. In \cite{Calandriello2011}, the security vulnerabilities of current vehicular communication architectures are identified. The work in \cite{Kim2010} proposed the use of multi-source filters to secure a vehicular network against data injection attacks (DIAs). Furthermore, the authors in \cite{Sun2017} introduced a new framework to improve the trustworthiness of beacons by combining two physical measurements (angle of arrival and Doppler effect) from received wireless signals. In \cite{2017Eltayeb}, the authors designed a multi-antenna technique for improving the physical layer security of vehicular millimeter-wave communications. Moreover, in \cite{PETRILLO2018}, the authors proposed a collaborative control strategy for vehicular platooning to address spoofing and denial of service attacks. The work in \cite{ferdowsi2018deep} developed a deep learning algorithm for authenticating sensor signals. Finally, an overview of current research on advanced intra-vehicle networks and the smart components of ITS is presented in \cite{Tuohy2015}.

In addition to cyber security in ITSs, physical safety and optimal control of ACVs have been studied in \cite{Kleberger2011,2015Bradley, XUE2014852, Sadraddini2017, Lefèvre2016, ferdowsidRL2018, 2016Schoitsch}. In \cite{Kleberger2011}, the authors identified the key vulnerabilities of a vehicle's controller and secured them using intrusion detection algorithms. The work in \cite{2015Bradley} analyzed the ACVs as cyber-physical systems and developed an optimal controller for their motion. The authors in \cite{XUE2014852} studied the safe operation of ACV networks in presence of an adversary that tries to estimate the dynamics of ACVs by its own observations. The authors in \cite{Sadraddini2017} proposed centralized and decentralized safe cruise control approaches for ACV platoons. A learning-based approach is proposed in \cite{Lefèvre2016} to control the velocity of ACVs. Furthermore, in \cite{ferdowsidRL2018}, the authors have proposed a robust deep reinforcement learning (RL) algorithm which mitigates cyber attacks on ACV sensors and maintains the safety of ACV system. In \cite{2016Schoitsch}, the authors studied the essence of secure and safe codesign for ACV systems.

However, despite their importance, the architecture and solutions in \cite{Woo2015,obd2016Narayanan,Calandriello2011,Kim2010,Sun2017,2017Eltayeb,PETRILLO2018,ferdowsi2018deep,Tuohy2015,Kleberger2011,2015Bradley, XUE2014852, Sadraddini2017, Lefèvre2016, ferdowsidRL2018, 2016Schoitsch} do not take into account the interdependence between the cyber and physical layers of ACVs while designing their security solutions. Moreover, the prior art in \cite{Woo2015,obd2016Narayanan,Calandriello2011,Kim2010,Sun2017,2017Eltayeb,PETRILLO2018,ferdowsi2018deep,Tuohy2015,Kleberger2011,2015Bradley, XUE2014852, Sadraddini2017, Lefèvre2016, ferdowsidRL2018, 2016Schoitsch}, does not provide solutions that can enhance the robustness of ACV motion control to malicious attacks. Nevertheless, designing an optimal and safe ITS requires robustness to attacks on intra-vehicle sensors as well as inter-vehicle communication. In addition, these existing works do not properly model the attacker's action space and goal (physical disruption in ITS) while providing their security solutions. In this context, the cyber-physical interdependence of the attacker's actions and goals will help providing better security solutions. Finally, the existing literature lacks a fundamental analysis of physical attacks as in the Jeep hijacking scenario\cite{Jeep2015Greenberg} in which the attacker aims at disrupting the ITS operation by causing non-optimality in a compromised ACV's speed.

The main contribution of this paper is, thus, a comprehensive study of joint cyber-physical security challenges and solutions in ACV networks which can be summarized as follows:
 
\begin{itemize}
	\item To address both safety and optimality of an ACV system, first an optimal safe controller is proposed so as to maximize the traffic flow and minimize the risk of accidents by optimizing the speed and spacing of ACVs. To the best of our knowledge, this work will be the first to analyze the physical attack on a ACV network and to prove that the proposed controller can maximize the stability and robustness of ACV systems against physical attacks such as in the Jeep hijacking scenario \cite{Jeep2015Greenberg}.
	\item  To improve the cyber-physical security study of ACV systems, next, new DIA detection approaches are proposed to address cyber attacks on ACV sensors and to analyze the physical impact of DIAs on an ACV system. To efficiently design the DIA detection approaches, ACV sensors are characterized in two subsets based on the availability of a priori information about their readings.
	\item For the first subset of sensors which have a priori information, a DIA detection approach is proposed derive an optimal threshold level for sensor errors which enables ACV to detect DIAs. For the second subset of sensors that lack a priori information, a novel multi-armed bandit (MAB) algorithm is proposed to learn which sensors are attacked. The proposed MAB algorithm uses the so-called Mahalanobis distance between the sensor data and an a-posteriori prediction to calculate a regret value and optimize the ACV's sensor fusion process by applying an upper confidence bound (UCB) algorithm. The proposed detection approaches maximize both the cyber security and physical robustness of ACV systems against DIAs.
\end{itemize}
  Simulation results show that the proposed optimal safe controller has higher safety, optimality, and robustness against physical attacks compared to other state-of-the-art approaches. In addition, our results show that the proposed DIA detection approaches yield an improved performance compared to Kalman filtering in mitigating the cyber attacks. Therefore, the proposed solutions improve the stability of ACV networks against DIAs.

The rest of the paper is organized as follows. Section \ref{sec:model} introduces our system model while Section \ref{sec:optimalcontroller} derives the optimal safe controller. Section \ref{sec:physicalattack} proves that the proposed optimal safe controller is robust against physical attacks.
Section \ref{sec:cyberattack} proposes approaches to mitigate cyber attacks on ACV while reducing the risk of accidents. Finally, simulation results are shown in Section \ref{sec:sim} and conclusions are drawn in Section \ref{sec:conc}.
\section{System Model}\label{sec:model}
\subsection{ACV Physical Model}
Consider an ACV, $ f $, that follows a leading ACV, $ l $ and tries to maintain a spacing from ACV $ l $ as shown in Fig. \ref{Fig:sysmodel}. Maintaining a spacing between ACVs is important to maximize the traffic flow and minimize the risk of accidents\cite{Kargl2008}. Let $ v_f $ be ACV $ f $'s speed in m/s. Then, ACV $ f $'s speed deviation can be written as $
	\dot{v}_f(t) = \frac{F_f(t)}{m_f} \triangleq u_f(t),$
where $ F_f(t) $ is $ f $'s engine force in Newtons (N), $ m_f $ is $ f $'s mass in kilograms (kg), and $ u_f(t) $ is $ f $'s physical controller input in N/kg . Moreover, letting $ v_l $ be $ l $'s speed, the spacing $ d(t) $ between $ f $ and $ l $ can be written as $	\dot{d}(t) = v_l(t) - v_f(t).$  Note that this model can be easily generalized to multiple ACVs by repeating the same set of equations for every pair of ACVs to capture any ACV network as shown in Fig. \ref{Fig:sysmodel}. 
\begin{figure}
	\centering
	\includegraphics[width=0.84\textwidth]{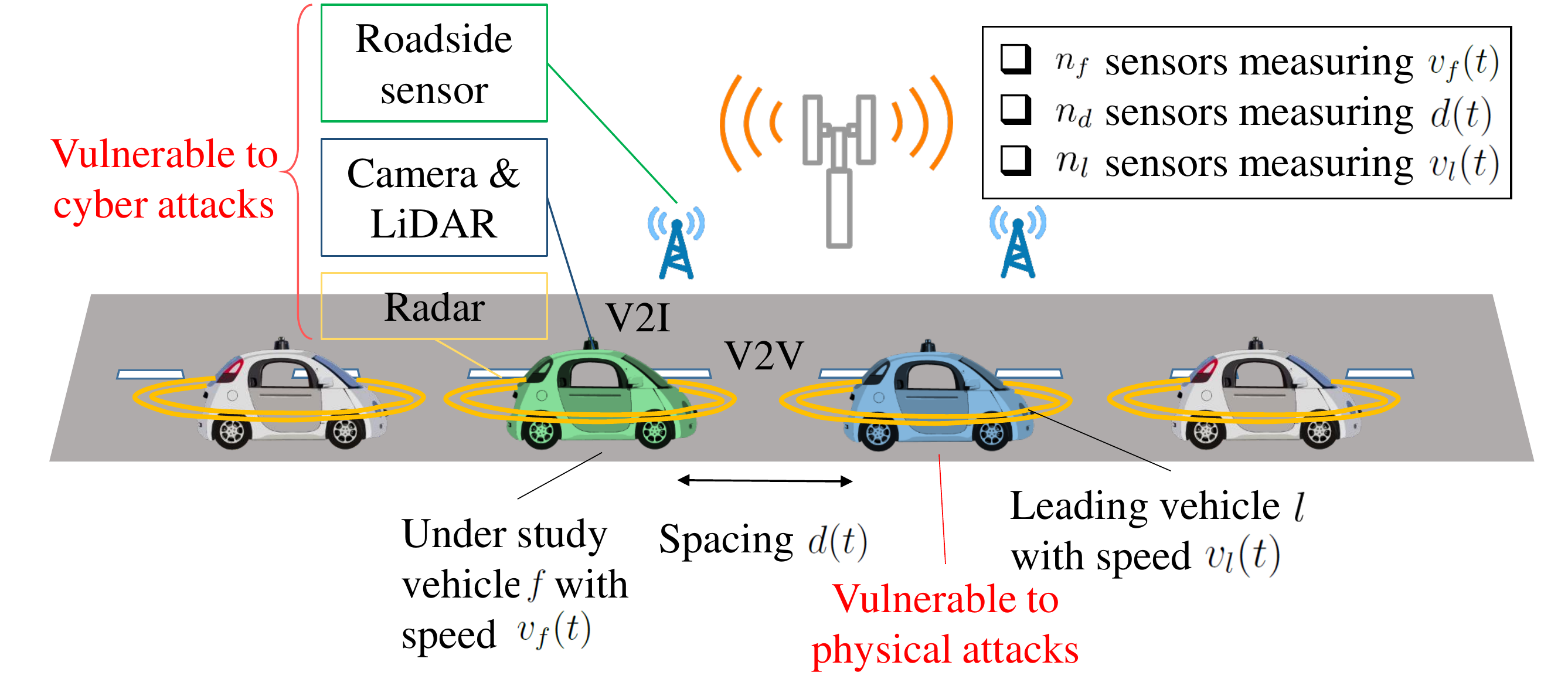}
	\vspace{-4mm}
	\caption{Illustration of the considered ACV system model.}
	\label{Fig:sysmodel}
	\vspace{-10mm}
\end{figure} 
Due to discrete time sensor readings in ACVs, we convert the aforementioned continuous system model to a discrete one using a linear transformation as follows\cite{williams2007linear}:
\begin{align}
	v_f(t+1) = v_f(t)+Tu_f (t),\,\,\,
	d(t+1) = d(t) + T v_l(t) - T v_f(t)\label{eq:spacingdisc},
\end{align}
where $ T $ is the sampling period of the sensors in seconds. The model can be summarized as:
\begin{align}\label{eq:system}
	\boldsymbol{x}(t+1) = \boldsymbol{A}\boldsymbol{x}(t)+\boldsymbol{B}u_f(t)+\boldsymbol{F}v_l(t) ,
\end{align}
where
\begin{align}
	\scriptstyle
		\boldsymbol{x}(t) = \left[\begin{array}{c}
		v_f(t) \\
		d(t)
		\end{array}\right],\,\,\boldsymbol{A}=\left[\begin{array}{c c}
		1 & 0 \\
		- T & 1
		\end{array}\right],\,\,\boldsymbol{B}=\left[\begin{array}{c}
		T \\
		0
		\end{array}\right],\,\,\boldsymbol{F}=\left[\begin{array}{c}
		0 \\
		T
		\end{array}\right].
\end{align}
To validate the practicality of the proposed system model, we need to show that $ u_f(t) $ can control the speed and spacing of $ f $ , i.e., $ u_f $ can take state vector $ \boldsymbol{x}(t_0) $ to any desired state $   \boldsymbol{x}(t_1)  $. The following remark shows that $ u_f $ can control system \eqref{eq:system}.
\begin{remark}
	If $ T>0  $, then the system in \eqref{eq:system} is controllable. To illustrate the reason, we know that
	the system \eqref{eq:system} is controllable if the rank of controllability matrix $ \boldsymbol{C} = \left[\begin{array}{c c}
	\boldsymbol{B} & \boldsymbol{A}\boldsymbol{B}
	\end{array}\right] $ is 2 (number of state variables)\cite{williams2007linear}. Thus, we have:
	\begin{align}
	\scriptstyle
		 \boldsymbol{C} = \left[\begin{array}{c c}
		 \begin{array}{c}
		 T \\
		 0
		 \end{array} &\left[\begin{array}{c c}
		 1 & 0 \\
		 - T & 1
		 \end{array}\right].\left[\begin{array}{c}
		 T \\
		 0
		 \end{array}\right]
		 \end{array}\right]=\left[\begin{array}{c c}
			T & T\\
			0 & -T^2
		 \end{array}\right].
	\end{align} 
	Therefore, for any $ T>0 $ the columns of $ \boldsymbol{C} $ are linearly independent which implies that the rank of $ \boldsymbol{C} $ will be $ 2 $.
\end{remark}
\subsection{ACV cyber model}\label{subsec:cybermodel}
In order to navigate, as shown in Fig. \ref{Fig:sysmodel}, ACV $ f $ relies on $ n_f $, $ n_l $, and $ n_d $ sensors which measure $ v_f(t) $, $ d(t) $, and $ v_l(t) $, respectively. For instance, multiple intra-vehicle inertial measurement units (IMUs) measure $ v_f(t) $, multiple cameras, radars, and LiDAR can measure $ d(t) $, and roadside sensors and ACV $ l $ measure $ v_l(t) $ and transmit the measurements to ACV $ f $ using vehicle-to-vehicle (V2V) and vehicle-to-infrastructure (V2I) communication links.
Thus, we model the sensor readings as follows:
\begin{align}\label{eq:sensors}
\boldsymbol{z}_f(t) = \boldsymbol{h}_f v_f(t) + \boldsymbol{e}_f(t),\,\, \boldsymbol{z}_l(t) = \boldsymbol{h}_l v_l(t) + \boldsymbol{e}_l(t),\,\,\boldsymbol{z}_d(t) = \boldsymbol{h}_d d(t) + \boldsymbol{e}_d(t),
\end{align} 
where $ \boldsymbol{z}_f(t) $, $ \boldsymbol{z}_l(t) $, and $ \boldsymbol{z}_d(t) $ are sensor vectors with $ n_f $, $ n_l $, and $ n_d $ elements which measure $ v_f(t) $, $ v_l(t) $, and $ d(t) $, respectively. Also, $ \boldsymbol{h}_f $, $ \boldsymbol{h}_l $, and $ \boldsymbol{h}_d $ are vectors with $ n_f $, $ n_l $, and $ n_d $ elements equal to $ 1 $. Moreover as assumed in \cite{williams2007linear} and \cite{nonlinear}, $ \boldsymbol{e}_f(t) $, $ \boldsymbol{e}_l(t) $, and $ \boldsymbol{e}_d(t) $ are noise vectors that follow a white Gaussian distribution with zero mean and variance vectors $ \boldsymbol{\sigma}_f^2 = \left[\sigma^2_{f_1},\dots,\sigma^2_{f_{n_f}}\right]^T $, $ \boldsymbol{\sigma}_l^2 = \left[\sigma^2_{l_1},\dots,\sigma^2_{l_{n_l}}\right]^T $, and $ \boldsymbol{\sigma}_d^2 = \left[\sigma^2_{d_1},\dots,\sigma^2_{d_{n_d}}\right]^T $. Since the sensor readings in \eqref{eq:sensors} are noisy, we need to optimally estimate $ v_f(t) $, $ v_l(t) $ and $ d(t) $ from $ \boldsymbol{z}_f(t) $, $ \boldsymbol{z}_l(t) $, and $ \boldsymbol{z}_d(t) $ by minimizing the estimation error. To find the optimal estimations $  \hat{v}_f(t) $, $  \hat{v}_l(t) $, and $  \hat{d}(t) $ (the estimations of $ v_f(t) $, $ v_l(t), $ and $ d(t) $), we use two types of estimators. A \emph{static estimator} to estimate the variables at the initial state, $ t_0 $, (the time step that $ f $ starts following $ l $) and a \emph{dynamic estimator} to estimate $  \hat{v}_f(t) $ and $  \hat{d}(t) $ using the state equations in \eqref{eq:spacingdisc}. We use a static estimator at the initial state because ACV $ f $ does not follow the speed of ACV $ l $ using \eqref{eq:spacingdisc} before $ t_0 $  and, hence, a dynamic estimator cannot be used due to lack of information about the dynamics of $ f $ before $ t_0 $. Moreover, for $ \hat{v}_l(t) $, we always use a static estimator since we do not have any information on the dynamics of $ v_l(t) $.

For the static estimator, we define a least-square (LS) cost function for each variable as follows:
\begin{align}
J_f(\hat{v}_f(t)) &= \frac{1}{2}\sum_{t=0}^{\infty}\left(\boldsymbol{z}_f(t)-\boldsymbol{h}_f\hat{v}_f(t)\right)^T\boldsymbol{R}_f^{-1}\left(\boldsymbol{z}_f(t)-\boldsymbol{h}_f\hat{v}_f(t)\right),
\\
J_l(\hat{v}_d(t)) &= \frac{1}{2}\sum_{t=0}^{\infty}\left(\boldsymbol{z}_l(t)-\boldsymbol{h}_l\hat{v}_l(t)\right)^T\boldsymbol{R}_l^{-1}\left(\boldsymbol{z}_l(t)-\boldsymbol{h}_l\hat{v}_l(t)\right),
\\
J_d(\hat{d}(t))& = \frac{1}{2}\sum_{t=0}^{\infty}\left(\boldsymbol{z}_d(t)-\boldsymbol{h}_d\hat{d}(t)\right)^T\boldsymbol{R}_d^{-1}\left(\boldsymbol{z}_d(t)-\boldsymbol{h}_f\hat{d}(t)\right),
\end{align}
where $ \boldsymbol{R}_f, \boldsymbol{R}_l $, and $ \boldsymbol{R}_d $ are the measurement covariance matrices associated with sensors measure $ v_f(t) $, $ v_l(t) $, and $ d(t) $, respectively. Moreover, since the sensors independently measuring the three variables, they will not have any noise covariance and $ \boldsymbol{R}_f, \boldsymbol{R}_l $, and $ \boldsymbol{R}_d $ will be diagonal. Therefore, explicit solution for the optimal estimation can be derived as:
\begin{align}
{\scriptsize \hat{v}_f(t)}&{\scriptsize= \left(\boldsymbol{h}_f^T\boldsymbol{R}_f^{-1}\boldsymbol{h}_f\right)^{-1}\boldsymbol{h}_f^T\boldsymbol{R}_f^{-1}\boldsymbol{z}_f(t)=\sum_{i=1}^{n_f} \underbrace{\frac{\frac{1}{\sigma_{f_i}^2}}{\sum_{i=1}^{n_f}\frac{1}{\sigma_{f_i}^2}}}_{w_{f_i}}z_{f_i}(t)=v_f(t)+\sum_{i=1}^{n_f}w_{f_i}e_{f_i}(t),}\\ {\scriptsize\hat{v}_l(t)}&{\scriptsize= \left(\boldsymbol{h}_l^T\boldsymbol{R}_l^{-1}\boldsymbol{h}_l\right)^{-1}\boldsymbol{h}_l^T\boldsymbol{R}_l^{-1}\boldsymbol{z}_l(t)=\sum_{i=1}^{n_l}\underbrace{ \frac{\frac{1}{\sigma_{l_i}^2}}{\sum_{i=1}^{n_l}\frac{1}{\sigma_{l_i}^2}}}_{w_{l_i}}z_{l_i}(t)=v_l(t)+\sum_{i=1}^{n_l}w_{l_i}e_{l_i}(t),}\\{\scriptsize\hat{d}(t)}&{\scriptsize= \left(\boldsymbol{h}_d^T\boldsymbol{R}_d^{-1}\boldsymbol{h}_d\right)^{-1}\boldsymbol{h}_d^T\boldsymbol{R}_f^{-1}\boldsymbol{z}_d(t)=\sum_{i=1}^{n_d}\underbrace{ \frac{\frac{1}{\sigma_{d_i}^2}}{\sum_{i=1}^{n_d}\frac{1}{\sigma_{d_i}^2}}}_{w_{d_i}}z_{d_i}(t) =d(t)+\sum_{i=1}^{n_d}w_{d_i}e_{d_i}(t).}
\end{align}
For the dynamic estimator, we use a Kalman filter which uses the state equation in the estimation process. To this end, we define an output equation as follows:
\begin{align}
\boldsymbol{z}(t) \hspace{-0.5mm}=\hspace{-0.5mm}\boldsymbol{H}\boldsymbol{x}(t) \hspace{-0.5mm}+\hspace{-0.5mm} \boldsymbol{e}(t), \text{s.t. } \boldsymbol{H}\hspace{-0.5mm} =\hspace{-0.5mm} \left[\begin{array}{c c}
\boldsymbol{h}_f & \boldsymbol{0}_{n_f\times 1}\\
\boldsymbol{0}_{n_d\times 1} & \boldsymbol{h}_d
\end{array}\right], \boldsymbol{z}(t) = \left[\begin{array}{c}
\boldsymbol{z}_f(t)\\
\boldsymbol{z}_d(t)
\end{array}\right],\boldsymbol{e}(t) = \left[\begin{array}{c}
\boldsymbol{e}_f(t)\\
\boldsymbol{e}_d(t)
\end{array}\right].
\end{align}

Note that we cannot apply the dynamic estimator on $ \hat{v}_l(t) $ since we do not have a-priori information about the dynamics of $ l $. To dynamically estimate $ \hat{d}(t) $ and $ \hat{v}_f(t) $, we use an a priori estimation derived from the state equations as well as a weighted residual of output error to correct the a priori estimation as follows\cite{kailath1980linear}:
\begin{align}
\underbrace{\hat{\boldsymbol{x}}(t)}_{\text{estimation}} = \underbrace{\hat{\boldsymbol{x}}^{(-)}(t)}_{\text{a priori estimation}} + \underbrace{\boldsymbol{K}(t) \overbrace{\left[\boldsymbol{z}(t) - \boldsymbol{H}\hat{\boldsymbol{x}}^{(-)}(t)\right]}^{\text{residual}}}_{\text{correction}},
\end{align}
where $ \boldsymbol{K}(t) $ is the \emph{Kalman gain}. By defining an a posteriori error covariance matrix $ \boldsymbol{P}(t) = \E \left[\boldsymbol{r}(t)\boldsymbol{r}^T(t)\right] $, where $ \boldsymbol{r}(t) = \boldsymbol{x}(t) - \hat{\boldsymbol{x}}(t) $, we can find a $ \boldsymbol{K}(t) $ to minimize $ \text{trace}\left[\boldsymbol{P}(t)\right] = \E\left[\boldsymbol{r}^T(t)\boldsymbol{r}(t)\right] $. The solution for such $ \boldsymbol{K}(t) $ can be given by\cite{kailath1980linear}:
\begin{align}
\scriptsize\hat{\boldsymbol{x}}^{(-)}(t) &\scriptsize= \boldsymbol{A} \hat{\boldsymbol{x}}(t-1) + \boldsymbol{B}u_f(t-1) + \boldsymbol{F} \hat{v}_l(t-1),\label{eq:aprioristate}\\
\scriptsize\boldsymbol{P}^{(-)}(t) &\scriptsize\triangleq \boldsymbol{A}\boldsymbol{P}(t-1)\boldsymbol{A}^T,\label{eq:Pgainpre}\\
\scriptsize\boldsymbol{K}(t) &\scriptsize= \boldsymbol{P}^{(-)}(t) \boldsymbol{H}^T\left[\boldsymbol{H}\boldsymbol{P}^{(-)}(t)\boldsymbol{H}^T+\boldsymbol{R}\right]^{-1},\label{eq:Kgain}\\
\scriptsize\hat{\boldsymbol{x}}(t) &\scriptsize= \left[\boldsymbol{I} -\boldsymbol{K}(t) \boldsymbol{H} \right] \hat{\boldsymbol{x}}^{(-)}(t)+ \boldsymbol{K}(t)\boldsymbol{z}(t),\label{eq:Kalmanestimtion}\\
\scriptsize\boldsymbol{P}(t) &\scriptsize=\left[\boldsymbol{I} -\boldsymbol{K}(t) \boldsymbol{H} \right] \boldsymbol{P}^{(-)}(t).\label{eq:Pgain}
\end{align}
where \small $ \boldsymbol{R}\hspace{-0.5mm} =\hspace{-1mm}\left[\begin{array}{c c}
\boldsymbol{R}_f & \boldsymbol{0}_{n_f\times n_d}\\
\boldsymbol{0}_{n_d\times n_f} & \boldsymbol{R}_d
\end{array}\right]\hspace{-0.5mm}$\normalsize is a block diagonal matrix. As can be seen from \eqref{eq:Pgainpre}, \eqref{eq:Kgain}, and \eqref{eq:Pgain}, the update processes for $ \boldsymbol{P}(t) $ and $ \boldsymbol{K}(t) $ are independent of the states and controller. Thus, $ \boldsymbol{P}(t) $ and $ \boldsymbol{K}(t) $ will converge to constant matrices, $ \tilde{\boldsymbol{P}} $ and $ \tilde{\boldsymbol{K}} $.

For the studied system, we now define the cyber-physical security problems for ACV systems that we will study next. We will address three main problems: 1) What is the optimal safe ACV controller for the system in \eqref{eq:system} that minimizes the risk of accidents while maximizing the traffic flow on the roads? 2) Is the proposed optimal safe controller for ACV systems robust and stable against physical attacks? and 3) How to securely fuse the sensor readings to mitigate DIAs on ACVs and minimize the impact of such attacks on the control of ACVs? Addressing these problems is particularly important because the model in \eqref{eq:system} identifies the microscopic characteristics of an ACV network and, thus, to achieve large-scale security and safety in ACV networks we must secure every ACV against cyber-physical attacks.

Addressing these three problem requires a comprehensive study of the interdependencies between the cyber and physical characteristics of ACV systems. Such interdependent cyber-physical study helps to find the vulnerabilities of the ACV systems against both cyber and physical attacks. Thus, we can derive an optimal controller that minimizes the risk of accidents and we can design cyber attack detection approaches that not only take into account the cyber characteristics of the ACV system, but also aims at minimizing the likelihood of collisions in ACV networks. Unlike the works in \cite{Woo2015,obd2016Narayanan,Calandriello2011,Kim2010,Sun2017,2017Eltayeb,PETRILLO2018,ferdowsi2018deep,Tuohy2015}, we consider the physical characteristics of ACV systems while developing DIA detection approaches. Moreover, the combined optimal and safe ACV controller design has not been studied previously in \cite{Kleberger2011,2015Bradley, XUE2014852, Sadraddini2017, Lefèvre2016, ferdowsidRL2018, 2016Schoitsch}.

\section{Optimal Safe ACV Controller} \label{sec:optimalcontroller}
Our first task is thus to derive an optimal safe controller for ACV systems. To analyze ACV $ f $'s optimal control input $ u_f(t) $, we define an \emph{optimal safe spacing} value as $ o(v_f(t)) \triangleq \frac{v_f^2(t)}{2b_fT} $, where $ b_f $ is ACV $ f $'s maximum braking deceleration. This value is defined so as to guarantee that if the leading ACV $l$ stops suddenly, the following ACV $f$ will stop completely before hitting ACV $ l $ as long as $ f $ starts braking immediately after observing $ l $'s braking process. This can be captured by the following energy equivalence condition:
\begin{align}
	\underbrace{\frac{m_f}{2}(v_f^2(t) - 0 )}_{\text{Kinetic Energy}} = \underbrace{m_fb_fd(t)}_{\text{Potential Energy}} \Rightarrow o(v_f(t)) = \frac{v_f^2(t)}{2b_f}.
\end{align} 
Our goal here is to maintain an optimal safe spacing between ACVs $ f $ and $ l $. Thus, we define a \emph{physical regret} $ R(t) $ as the square of difference between the optimal safe spacing and the actual spacing. This regret quantifies the safety and optimality of ACV motion by preventing any collisions and minimizing the spacing between the ACVs and can be written as follows:
\begin{align}
	R(t) &=  \left(o(v_f(t+1)) - d(t+1)\right)^2 = \left(\frac{v_f^2(t+1)}{2b_f} - d(t+1)\right)^2\nonumber\\\label{eq:costfunction}
	&=\left(\frac{({v}_f(t)+Tu_f(t))^2}{2b_f}-{d}(t)-T{v}_l(t)+T{v}_f(t)\right)^2
\end{align}
In addition, each ACV only have access to estimation of  $ v_f(t) $, $ v_l(t) $, and $ d(t) $. Thus, ACV $ f $ must design an input $ u_f(t) $ to minimize an estimation of physical regret which is defined as follows:
\begin{align}
	\hat{R}(t)&= \left(\frac{(\hat{v}_f(t)+Tu_f(t))^2}{2b_f}-\hat{d}(t)-T\hat{v}_l(t)+T\hat{v}_f(t)\right)^2.
\end{align}
This problem is challenging to solve because $ \hat{v}_l(t) $ is an independent parameter and ACV $ f $ cannot be sure about the future values of $ \hat{v}_l(t) $. To solve this problem, we consider two scenarios: a) ACV $ f $ has no prediction about ACV $ l $'s future speed values (\emph{One-step ahead controller}) and b) ACV $ f $ has a predictor which can predict ACV $ l $'s future speed value for $ N $ time steps (\emph{$ N $-step ahead controller})(such predictors have attracted recent attention in the transportation literature, e.g., see \cite{Tian2015} and \cite{MA2015187}). Next, we propose an optimal controller for these two cases.
\subsection{One-step ahead controller}

To solve the one-step ahead controller problem, we consider some physical limitations on the speed and the control input. We prohibit the speed from being greater than the free-flow speed of a road, $ \tilde{v} $. Moreover, due to the physical capabilities of the vehicle and for maintaining passengers' comfort, we must have a limitation on the control input and speed deviation. Thus, the optimization problem of the ACV $ f $ can be written as follows:
\begin{align}
	u^*_f(t) &= \argmin_{u_f(t)} \hat{R}(t)\label{eq:1step}\\
	\text{s.t.}\,\,\,\,\,& u_f^{\text{min}}\leq u_f(t)\leq u_f^{\text{max}},\label{eq:1stepconst1} \\
	&0\leq v_f(t+1)\leq\tilde{v}\label{eq:1stepconst2},\\
	& |u_f(t)-u_f(t-1)| \leq \Delta u, \label{eq:1stepconst3}
\end{align}
where $ u_f^{\text{min}}<0 $ and $ u_f^{\text{max}}>0 $ are the minimum and maximum allowable control input and $ \Delta u $ is the maximum allowable change in the controller to yield a comfortable ride.
\begin{theorem}\label{Theorem1step}
	The one-step ahead optimal controller is:
	\begin{align}\label{eq:1stepopt}
		u_f^*(t) = 
		\min\left\{\max\left\{\underbar{$ u $}_1(t)  , \frac{1}{T}\left(\sqrt{2b_f\left(\hat{d}(t)+T\hat{v}_l(t)-T\hat{v}_f(t)\right)}-\hat{v}_f(t)\right)\right\},\overline{u}_1(t)\right\}
	\end{align}
	where $\underbar{$ u $}_1(t)\hspace{-0.5mm} \triangleq\hspace{-0.5mm} \max\hspace{-0.5mm}  \left\{\hspace{-0.5mm} \frac{-\hat{v}_f(t)}{T}\hspace{-0.5mm} ,\hspace{-0.5mm} u_f^{\text{min}}\hspace{-0.5mm},\hspace{-0.5mm}u_f(t-1)\hspace{-0.5mm}-\hspace{-0.5mm}\Delta u \right\} $ and $ \overline{u}_1(t) \triangleq \min \Big\{\frac{\tilde{v}-\hat{v}_f(t)}{T},u_f^{\text{max}},u_f(t-1)+ \Delta u \Big\} $.
\end{theorem}
\begin{proof}
	See Appendix \ref{App:Theorem1}.
\end{proof}
Theorem \ref{Theorem1step} derives the optimal controller for the ACV $ f $ when it only optimizes its action for the next step without considering future actions. In the next subsection, we derive the ACV $ f $'s optimal controller when it considers minimizing the regret for $ N $ step ahead.
\subsection{$ N $-step ahead controller}
To find the $ N $-step ahead controller, first we define a discount factor $ 0\leq\gamma\leq 1 $ which specifies the level of future physical regret for the decision taken at each time step. Thus, by defining the $ N $-step ahead total discounted physical regret $ \mathcal{R}(t,N) \triangleq \sum_{\tau = t } ^{t+N-1} \gamma^{\tau - t} \hat{R}(\tau), $ the controller optimization problem can be written as follows:
\begin{align}
	u_f^*(t)&=[1,\underbrace{0,\dots,0}_{N-1}]\Bigg\{\argmin_{\boldsymbol{u}^N_f(t)}\mathcal{R}(t,N)\Bigg\},&\label{eq:minimization}
\end{align}
where the conditions in \eqref{eq:1stepconst1}, \eqref{eq:1stepconst2}, and \eqref{eq:1stepconst3} hold true. Moreover $ \boldsymbol{u}^N_f(t) = [u_f(t),\dots,u_f(t+N-1)]^T $, $ N $ is the number of future steps which is taken into account in finding the optimal controller, and $ u_f^*(t) $ is the optimal controller at time step $ t $.

\begin{theorem}\label{TheoremNstep}
	The solution of $ N $-step ahead controller is equivalent to the one-step ahead controller.
\end{theorem}
\begin{proof}
To solve the problem in \eqref{eq:minimization}, we use a so called indirect method. To this end, we start by defining an \emph{augmented physical regret} using the following state equation:
\begin{align}
	\mathcal{R}'(t,N) &= \mathcal{R}(t,N) + \sum_{\tau = t}^{t+ N -1} \boldsymbol{\lambda}^T(\tau+1)\left[\underbrace{\boldsymbol{A}\hat{\boldsymbol{x}}(\tau)+\boldsymbol{B}\boldsymbol{u}_f(\tau)+\boldsymbol{F}v_l(\tau)}_{\boldsymbol{g}(\tau)}-\hat{\boldsymbol{x}}(\tau+1)\right]\nonumber\\
	& = \sum_{\tau = t}^{t+ N -1}\left[\gamma^{\tau-t}\hat{R}(\tau) + \boldsymbol{\lambda}^T(\tau+1)\left[\boldsymbol{g}(\tau)-\hat{\boldsymbol{x}}(\tau+1)\right]\right],\label{eq:augmented}
\end{align}
where $ \boldsymbol{\lambda}(\tau) = \left[\lambda_1(\tau),\lambda_2(\tau)\right]^T $.
Then, let \emph{Hamiltonian} function defined as $ \mathcal{H}(\tau) \triangleq \gamma^{\tau-t}\hat{R}(\tau) + \boldsymbol{\lambda}^T(\tau+1)\boldsymbol{g}(\tau)$. Thus, we can write \eqref{eq:augmented} as follows:
\begin{align}
	\mathcal{R}'(t,N) &= \underbrace{\boldsymbol{\lambda}^T(t+N)\hat{\boldsymbol{x}}(t+N)}_{\text{terminal time}} + \underbrace{\mathcal{H}(t)}_{\text{initial time}} + \underbrace{\sum_{\tau = t+1}^{t+ N -1} \left[\mathcal{H}(\tau) - \boldsymbol{\lambda}^T(\tau)\hat{\boldsymbol{x}}(\tau) \right]}_{\text{running time}}.
\end{align}
Thus, to find critical points (candidate minima) we must solve $ \nabla\mathcal{R}'(t,N) = 0  $. First, we start by finding the differential $ d \mathcal{R}'(t,N) $ and then we identify the derivatives as follows:
\begin{align}
	d \mathcal{R}'(t,N) &=  \boldsymbol{\lambda}^T(t+N) d\hat{\boldsymbol{x}}(t+N) + \left(\nabla_{\hat{\boldsymbol{x}}(t)}\mathcal{H}(t)\right)^Td\hat{\boldsymbol{x}}(t)+ \sum_{\tau=t+1}^{t+N-1}\left(\nabla_{\boldsymbol{x}(\tau)}\mathcal{H}(\tau)-\boldsymbol{\lambda}(\tau)\right)^Td \hat{\boldsymbol{x}}(\tau)\nonumber\\
	&+\sum_{\tau=t}^{t+N-1}\left(\nabla_{\boldsymbol{u}(\tau)}\mathcal{H}(\tau)\right)^Td\boldsymbol{u}(\tau) +  \sum_{\tau=t+1}^{t+N}\left(\nabla_{\boldsymbol{\lambda}(\tau)}\mathcal{H}(\tau-1)-\hat{\boldsymbol{x}}(\tau)\right)^Td \boldsymbol{\lambda}(\tau).
\end{align}
Thus, to have $ d \mathcal{R}'(t)=0 $ each of the terms in brackets must be equal to zero:
\begin{align}
	&\boldsymbol{\lambda}(t+N) = 0,\\
	&\nabla_{\hat{\boldsymbol{x}}(t)}\mathcal{H}(t)  = 0,\Rightarrow\hat{\boldsymbol{x}}(t)=\left[\hat{v}_f(t),\hat{d}(t)\right]^T,\\
	&\boldsymbol{\lambda}(\tau)=\nabla_{\hat{\boldsymbol{x}}(\tau)}\mathcal{H}(\tau) = \nabla_{\hat{\boldsymbol{x}}(\tau)} \gamma ^{t-\tau}\hat{R}(\tau) +  \nabla_{\hat{\boldsymbol{x}}(\tau)}\boldsymbol{g}(\tau)\boldsymbol{\lambda}(t+1), & \forall \tau =  t+1,\dots,t+N-1,\label{eq:costate}\\
	&0=\nabla_{\boldsymbol{u}(\tau)}\mathcal{H}(\tau)  = \nabla_{\boldsymbol{u}(\tau)} \gamma ^{t-\tau}\hat{R}(\tau) +  \nabla_{\boldsymbol{u}(\tau)}\boldsymbol{g}(\tau)\boldsymbol{\lambda}(t+1), & \forall \tau =  t,\dots,t+N-1,\label{eq:FOCHam}\\
	&\nabla_{\boldsymbol{\lambda}(\tau +1)}\mathcal{H}(\tau) = \hat{\boldsymbol{x}}(\tau +1) \Rightarrow \hat{\boldsymbol{x}}(\tau +1) = \boldsymbol{g}(\tau), & \forall \tau =  t,\dots,t+N-1.\label{eq:stateeq}
\end{align}
Now, using \eqref{eq:costate} we will have:
\begin{align}
	{\textstyle\boldsymbol{\lambda}(\tau) = \left[\begin{array}{c}
	\frac{\hat{v}_f(\tau)+Tu_f(\tau)}{b_f}+T\\
	-1
	\end{array}\right] \gamma^{t-\tau}\left(\frac{(\hat{v}_f(\tau)+Tu_f(\tau))^2}{2b_f}-\hat{d}(\tau)-T\hat{v}_l(\tau)+T\hat{v}_f(\tau)\right)
	+\boldsymbol{A}^T\boldsymbol{\lambda}(\tau+1).}\label{eq:costatesimp}
\end{align}
Moreover, from \eqref{eq:FOCHam} we will have:
\begin{align}
	0 & = \gamma^{t-\tau}\left(\frac{(\hat{v}_f(\tau)+Tu_f(\tau))^2}{2b_f}-\hat{d}(\tau)-T\hat{v}_l(\tau)+T\hat{v}_f(\tau)\right)\frac{T(\hat{v}_f(\tau)+Tu_f(\tau))}{b_f}+\boldsymbol{B}^T\boldsymbol{\lambda}(\tau+1).\label{eq:Ulambda}
\end{align}
Since $ \boldsymbol{\lambda}(t+N) = 0 $, then from \eqref{eq:Ulambda}, we derive:
\begin{align}
{\textstyle
	 u_f( t+N-1)\hspace{-0.5mm}=\hspace{-0.5mm}\frac{1}{T}\hspace{-0.5mm}\Bigg[\hspace{-1mm}\pm\hspace{-1mm}\sqrt{\hspace{-0.5mm}2b_f\hspace{-0.5mm}\left(\hspace{-0.5mm}\hat{d}(t\hspace{-0.5mm}+\hspace{-0.5mm}N\hspace{-0.5mm}-\hspace{-0.5mm}1)\hspace{-0.5mm}+\hspace{-0.5mm}T\hat{v}_l(t+N-1)\hspace{-0.5mm}-\hspace{-0.5mm}T\hat{v}_f(t+N-1)\hspace{-0.5mm}\right)}
	 \hspace{-0.5mm}-\hspace{-0.5mm}\hat{v}_f(t+N-1)\hspace{-0.5mm}\Bigg]}.
\end{align}
By substituting $ u_f( t+N-1) $ in \eqref{eq:costatesimp}, we obtain $ \lambda(t+N-1)= 0 $. This process can continue until $ \tau = t $ where we obtain $ \boldsymbol{\lambda}(t+1) = 0 $ and $ u_f( t)=\frac{1}{T}\Bigg[\pm\sqrt{2b_f\left(\hat{d}(t)+T\hat{v}_l(t)-T\hat{v}_f(t)\right)}-\hat{v}_f(t)\Bigg] $. Moreover, considering the constraints \eqref{eq:1stepconst1}, \eqref{eq:1stepconst2}, and \eqref{eq:1stepconst3}, we will end up having the optimal controller as defined in Theorem \ref{Theorem1step}. Thus, we prove that the $ N $-step ahead optimal controller is equivalent to $ 1 $-step ahead optimal controller. 
\end{proof}
From Theorem \ref{TheoremNstep}, we can observe that, if the ACV $ f $ minimizes its immediate physical regret, it will also minimize its long-term physical regret. This result shows that the proposed optimal safe controller does not require any information from future dynamics of ACV $ l $ as done in \cite{Tian2015} and \cite{MA2015187}.

\section{Physical Attack on ACV systems}\label{sec:physicalattack}
As derived in the previous section, the proposed optimal controller is a function of $ \hat{{v}}_l(t) $. This makes ACV $ f $ vulnerable against a physical attack on ACV $ l $. Thus, we now analyze whether an attacker can cause instable dynamics at ACV $ f $ by controlling $ l $. Consider an adversary who takes the control of ACV $ l $ and tries to cause instability in ACV $ f $'s speed, $ v_f(t) $, and spacing $ d(t) $. Using our derived optimal controller, by ensuring $ u_f(t) $ is not saturated ($ \underbar{$u$}_1(t)\leq u_f(t) \leq \overline{u}_1(t) $), and considering the estimated values to be close to the real values we will have:
\begin{equation}
\begin{rcases}
v_f(t+1) = \underbrace{\sqrt{2b_f(d(t)+Tv_l(t)-Tv_f(t)}}_{g_1(\boldsymbol{x}(t),v_l(t))},\\
d(t+1) = \underbrace{d(t) + Tv_l(t) - Tv_f(t)}_{g_2(\boldsymbol{x}(t),v_l(t))}.
\end{rcases} \Rightarrow\label{eq:optimal}\boldsymbol{x}(t+1) = \boldsymbol{g}(\boldsymbol{x}(t),v_l(t)),
\end{equation}
where $ \boldsymbol{g} = \left[g_1(\boldsymbol{x}(t),v_l(t)),g_2(\boldsymbol{x}(t),v_l(t))\right]^T $. 
\eqref{eq:optimal} is designed such that ACV $ f $ will always maintain an optimal safe spacing with $ l $. However, from \eqref{eq:optimal} we can see that the behavior of the system is a function of $ v_l $. Thus, next, we will analyze the physical attack scenario that is analogous to the Jeep hijacking case in\cite{Jeep2015Greenberg}.

\subsection{Stability Analysis}
To analyze the stability of \eqref{eq:optimal}, first, we define some useful concepts.
\begin{definition}\label{def:equlibrium}
	\cite{nonlinear} $ \bar{\boldsymbol{x}}(v_l) $ is said to be an \emph{equilibrium} point for the system in \eqref{eq:optimal} and a constant input $ v_l $, if $ \boldsymbol{g}(\bar{\boldsymbol{x}}(v_l),v_l) = \bar{\boldsymbol{x}}(v_l) $.
\end{definition}
An equilibrium indicates a point at which the states will not change. From Definition \ref{def:equlibrium} we can derive the equilibrium point for \eqref{eq:optimal} by considering a constant input $ v_l $ and solving the following set of equations:
\begin{align}
\bar{v}_f(v_l) = \sqrt{2b_f(\bar{d}(v_l)+Tv_l-T\bar{v}_f(v_l)},\,\,\,
\bar{d}(v_l) = \bar{d}(v_l) + Tv_l - T\bar{v}_f(v_l),
\end{align}
which results in: $
\bar{\boldsymbol{x}}(v_l) =  \left[v_l,\frac{v_l}{2b_f}\right]^T .$ The derived value for $ \bar{\boldsymbol{x}}(v_l) $ shows that in order to reach an equilibrium, ACV $ f $ must maintain the optimal safe spacing from ACV $ l $ and its speed must equal to $ v_l $. Thus, next, we show that our derived optimal controller is robust, i.e., under our controller if an adversary hijacks ACV $ l $, it cannot cause instability in $ v_f(t) $ and $ d(t) $.

\begin{definition}\label{def:stability}
	A system is called \emph{asymptotically stable} around its equilibrium point if it satisfies the following two conditions\cite{nonlinear}: 1) Given any $ v_l>0 $ and $ \varepsilon>0 $, $ \exists \delta_1>0 $ such that if $ |\boldsymbol{x}(t_s)-\bar{\boldsymbol{x}}(v_l)| $, then $ |\boldsymbol{x}(t)-\bar{\boldsymbol{x}}(v_l)|<\epsilon $, $ \forall t>t_s $ and 2) $ \exists \delta_2 > 0 $ such that if $ |x(t_s)-\bar{\boldsymbol{x}}(v_l)|<\delta_2 $, then $ \boldsymbol{x}(t) \rightarrow \bar{\boldsymbol{x}}(v_l) $ as $ t  \rightarrow \infty $.
\end{definition}
The first condition requires the state trajectory to be confined to an arbitrarily small ``ball" centered at the equilibrium point and of radius $ \varepsilon $, when released from an arbitrary initial condition in a ball of sufficiently small (but positive) radius $ \delta_1 $. This is called stability in the \emph{Lyapunov} sense\cite{nonlinear,Kasgari2018Stochastic}. It is possible to have Lyapunov stability without
having asymptotic stability. 

Next, we show how a Lyapunov function can help to analyze the stability of system \eqref{eq:optimal}\cite{nonlinear}. From \cite{nonlinear}, we know that if there exists a Lyapunov function for system \eqref{eq:optimal}, then $ \boldsymbol{x}(t) = \bar{\boldsymbol{x}}(v_l) $ is a stable equilibrium point in the sense of Lyapunov. In addition, if $ L(\boldsymbol{g}(\boldsymbol{x}(t),v_l))- L(\boldsymbol{x}(t))< 0$ then $ \boldsymbol{x}(t) = \bar{\boldsymbol{x}}(v_l) $ is an asymptotically stable equilibrium point. We can prove that the system in \eqref{eq:optimal} is asymptotically stable for the equilibrium point $ \bar{\boldsymbol{x}}(v_l) $, as follows.
\begin{proposition}\label{Proposition:Stability}
	$ \boldsymbol{x}(t) = \bar{\boldsymbol{x}}(v_l) $ is a stable equilibrium point in the Lyapunov sense.
\end{proposition}
\begin{proof}
	Let $ L(\boldsymbol{x}(t)) = \left(\frac{v_f^2(t)}{2b_f}-d(t)\right)^2 $. Then we will have $ L(\bar{\boldsymbol{x}}(v_l)) =0 $. Moreover, we can show that $  L(\boldsymbol{x}(t))\geq 0 ,$ $ \forall \boldsymbol{x}(t)$. Now, to check if $ L $ is a Lyapunov function we have:
	\begin{align}
	L(\boldsymbol{g}(\boldsymbol{x}(t),v_l))- L(\boldsymbol{x}(t)) &= \left(\frac{1}{2b_f}\left(\sqrt{2b_f(d(t)+Tv_l-Tv_f(t)}\right)^2-d(t) + Tv_l - Tv_f(t)\right)^2\nonumber\\
	&-\left(\frac{v_f^2(t)}{2b_f}-d(t)\right)^2=-\left(\frac{v_f^2(t)}{2b_f}-d(t)\right)^2 \leq 0.
	\end{align}
	Thus, $ \boldsymbol{x}(t) = \bar{\boldsymbol{x}}(v_l) $ is a stable equilibrium point in the sense of Lyapunov.
\end{proof}

From Proposition \ref{Proposition:Stability}, we can see that, as long as $ f $ follows $ l $ using our proposed controller, its speed and spacing from $ l $ will stay stable and will not be affected by the physical attack on $ l $. This shows that, not only our proposed controller maximizes the safety and optimality in ITS roads, but also it is robust to physical attacks such as in the Jeep scenario \cite{Jeep2015Greenberg}.

However, as can be seen from Theorems \ref{Theorem1step} and \ref{TheoremNstep}, even though it is robust to physical attacks, the derived optimal controller is largely dependent on the estimated values $ \hat{v}_f(t) $, $ \hat{v}_l(t) $, and $ \hat{d} $. Thus, an adversary can manipulate the sensor data to inject error in the estimation and ultimately increase the ACV $ f $'s physical regret. Analyzing such attacks require a cyber-physical study of the ACV system to derive approaches that mitigate attacks on sensors and minimize the effect of such attacks on the physical regret. Thus, next, we analyze the cyber attack on the ACV system.

\section{Cyber Attack on ACV Systems}\label{sec:cyberattack}
We now consider a cyber attacker that injects faulty data to sensor readings (cameras, LiDARs, radars, IMUs, and roadside sensors) such that the attacked sensor vector can be written as:
\begin{align}
	\left[\begin{array}{c}\bar{\boldsymbol{z}}(t)\triangleq\left[\begin{array}{c}
	\bar{\boldsymbol{z}}_f(t)\\
	\bar{\boldsymbol{z}}_d(t)
	\end{array}
	\right]\\
	\bar{\boldsymbol{z}}_l(t)
	\end{array}\right]\triangleq \left[\begin{array}{c}
	\boldsymbol{z}_f(t)\\
	\boldsymbol{z}_d(t)\\
	\boldsymbol{z}_l(t)
	\end{array}\right] + \left[\begin{array}{c}
	\boldsymbol{a}_f(t)\triangleq[a_{f_1},\dots,a_{f_{n_f}}]^T\\\boldsymbol{a}_d(t)\triangleq[a_{d_1},\dots,a_{d_{n_d}}]^T\\\boldsymbol{a}_l(t)\triangleq[a_{l_1},\dots,a_{l_{n_l}}]^T
	\end{array}\right] ,
\end{align}
where $ \boldsymbol{a}_f(t) $, $ \boldsymbol{a}_d(t) $, and $ \boldsymbol{a}_l(t) $ are data injection vectors on sensors and $ \bar{\boldsymbol{z}}_f(t) $, $ \bar{\boldsymbol{z}}_d(t) $, and $ \bar{\boldsymbol{z}}_l(t) $ are compromised sensor readings from $ v_f(t) $, $ d(t) $, and $ v_l(t) $, respectively. As we discussed in the state estimation section, we have a-priory information about sensors which measure $ v_f(t) $ or $ d(t) $, but such information is lacking for sensors that collect data from $ v_l(t) $. Thus, next, we consider cyber attacks on sensors a) with a-priori information and b) without a-priori information.
\subsection{Attack on sensors with a-priori information}
As discussed in Subsection \ref{subsec:cybermodel}, we use Kalman filtering to estimate $ \hat{v}_f(t) $ and $ \hat{d}(t)$. However, Kalman filtering is not robust to DIAs \cite{Yang2013}. Thus, we propose a filtering mechanism that can limit the effect of the DIA on $ \hat{v}_f(t) $ or $ \hat{d}(t)$. To this end, we use the a priori estimation $ \hat{\boldsymbol{x}}^{(-)}(t) $ at each time step to find an a priori sensor reading $ \boldsymbol{z}^{(-)}(t)\triangleq\boldsymbol{H}\hat{\boldsymbol{x}}^{(-)}(t) $. Then, the attack detection filter checks the absolute value of the residual $ \left|\boldsymbol{\mu}(t)\right| \triangleq \Bigg[|\mu_{{f}_1}(t)|,\dots, |\mu_{{f}_{n_f}}(t)|, |\mu_{{d}_1}(t)|,\dots, |\mu_{{d}_{n_d}}(t)|\Bigg]^T\triangleq\left|{\boldsymbol{z}}(t) - \boldsymbol{z}^{(-)}(t)\right|<\boldsymbol{\eta} $, where $ \boldsymbol{\eta} $ is the threshold vector. Any sensor which violates this inequality will be considered as a compromised sensor and will not be involved in Kalman filter update procedure. To find an optimal value for the threshold $ \boldsymbol{\eta} $, we next characterize the stochastic behavior of residual $ \boldsymbol{\mu}(t) $ when the ACV is not under attack. 

\begin{theorem}\label{Theorem:mudistribtution}
	The residual $ \boldsymbol{\mu}(t) $ follows a Gaussian distribution with zero mean and covariance matrix as follows:
	\begin{align}\label{eq:covmu}
	\boldsymbol{C}_{\mu} =  \boldsymbol{H}\boldsymbol{C}_r\boldsymbol{H}^T + \boldsymbol{R} -\boldsymbol{R}\tilde{\boldsymbol{K}}^T\boldsymbol{Q}^T\boldsymbol{H}^T - \boldsymbol{H}\boldsymbol{Q}\tilde{\boldsymbol{K}}\boldsymbol{R},
	\end{align}
	and $ \boldsymbol{C}_r $ is the solution of following discrete Ricatti equation $
	\boldsymbol{C}_{r} = \boldsymbol{Q}\boldsymbol{A}\boldsymbol{C}_r\boldsymbol{A}^T\boldsymbol{Q}^T + \boldsymbol{C}_{\rho}, $
	where $
		\boldsymbol{C}_{\rho} = \boldsymbol{Q}\boldsymbol{F}\sigma^2_l\boldsymbol{F}^T\boldsymbol{Q}^T + \boldsymbol{Q}\tilde{\boldsymbol{K}}\boldsymbol{R}\tilde{\boldsymbol{K}}^T\boldsymbol{Q}^T,$ and $
		\boldsymbol{Q} = \left[\boldsymbol{I}+\left[\boldsymbol{I} -\tilde{\boldsymbol{K}} \boldsymbol{H} \right]^{-1}\tilde{\boldsymbol{K}}\boldsymbol{H}\right]^{-1}.$
\end{theorem}
\begin{proof}
	See Appendix \ref{App:Theoremmu}.
\end{proof}
Theorem \ref{Theorem:mudistribtution} derives the distribution of $ \boldsymbol{\mu}(t)$ which we will use next to find an optimal value for the threshold level $ \boldsymbol{\eta} $. Fig. \ref{Fig:Dist} and Fig. \ref{Fig:Cmu} show a comparison between simulation and analytical results derived for the mean and covariance matrix of $ \boldsymbol{r}(t) $, $ \rho(t) $, and $ \mu(t) $. From Figs. \ref{Fig:Dist} and \ref{Fig:Cmu} we can that the analytical results match the simulation results which validates Theorem \ref{Theorem:mudistribtution}.
\begin{figure}[t]
		\centering
		\includegraphics[width=0.65\textwidth]{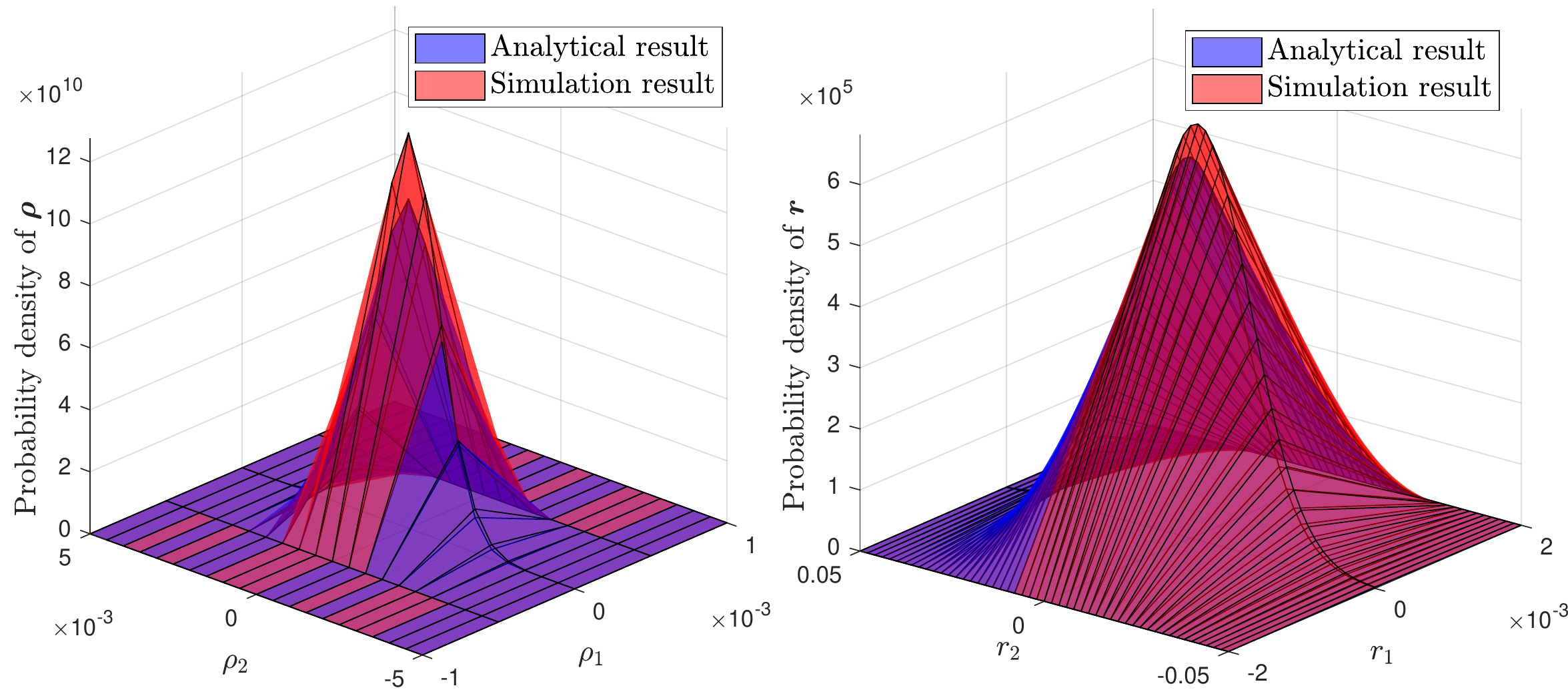}
		\vspace{-4mm}
		\caption{Analytical and simulation result for cumulative density function of $ \boldsymbol{\rho}(t) $ and $ \boldsymbol{r}(t) $.}
		\label{Fig:Dist}\vspace{-6mm}
\end{figure}
\begin{figure}[t]
	\centering
	\includegraphics[width=0.65\textwidth]{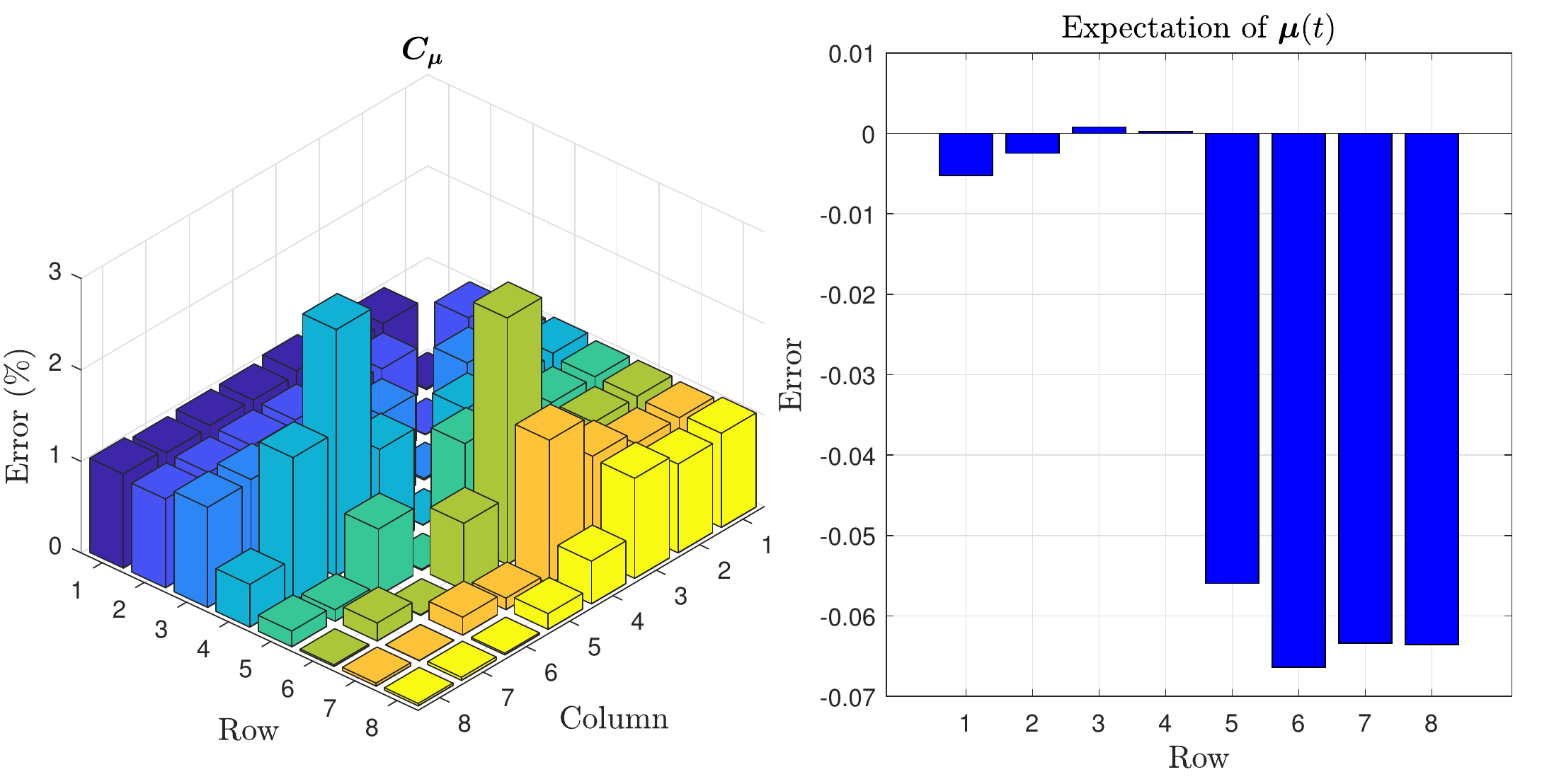}
	\vspace{-4mm}
	\caption{Analytical and simulation result for the mean and variance of $ \boldsymbol{\mu}(t) $.}
	\label{Fig:Cmu}
	\vspace{-10mm}
\end{figure}

We can now find the probability with which $ \mu_i(t) $ (element $i$ of $\boldsymbol{\mu}(t)$) remains below the threshold value $ \eta_i $ as $
	\text{Pr}\left(|\mu_i|\leq \eta_i\right) = \text{Pr}\left(-\eta_i\leq\mu_i\leq\eta_i\right) = \Phi_{\mu_i}(\eta_i)- \Phi_{\mu_i}(-\eta_i),$
where $ \Phi_{\mu_i}(.) $ is the cumulative density function of $ \mu_i $ which follows a Gaussian distribution with zero mean and variance $ {C}_{\mu}(i,i) $ (element in $ i $-th row and $ i $-th column of $ \boldsymbol{C}_{\mu} $). Thus, we can derive the optimal value for $ \boldsymbol{eta} $ by defining $ \text{Pr}\left(|\mu_i|\leq \eta_i\right) $ for every sensor. For instance, choosing values  $ \eta_i = {C}_{\mu}(i,i), 2{C}_{\mu}(i,i),$ or $ 3{C}_{\mu}(i,i) $ will result in $ \text{Pr}\left(|\mu_i|\leq \eta_i\right) = 0.68, 0.95, $ or $ 0.997 $. Even by defining the threshold value, the attacker might stay stealthy in some cases if it controls the amount of injected data to the sensors. Next, we find a relationship between the maximum value of DIA and the probability of staying stealthy.
\begin{proposition}\label{Proposition:probattack}
	 The probability with which an attack vector $ \tilde{\boldsymbol{a}} \triangleq \left[\boldsymbol{a}_f,\boldsymbol{a}_d\right]^T $ will not trigger the $ i $-th element of attack detection filter, $ \mu_i $, (stealthy attack) will be given by:
	 \begin{align}
	 	{p}_{i}({\tilde{\boldsymbol{a}}})\triangleq\textrm{Pr}\left(|\mu_i|\leq \eta_i|\tilde{\boldsymbol{a}}\right)= \Phi_{\mu_i}(\boldsymbol{\varPsi}_i\tilde{\boldsymbol{a}}+\eta_i)-\Phi_{\mu_i}(\boldsymbol{\varPsi}_i\tilde{\boldsymbol{a}}-\eta_i),\label{eq:attacksuccess}
	 	\end{align}
	 	where $ \boldsymbol{\varPsi}_i $ is the $ i $-th row of $
	 	\boldsymbol{\varPsi} = \left[ \boldsymbol{I}  - \boldsymbol{H}\left[\boldsymbol{I} - \boldsymbol{Q}\boldsymbol{A}\right]^{-1} \boldsymbol{Q}\tilde{\boldsymbol{K}}\right].$
\end{proposition}
\begin{proof}
	Suppose that the attacker initiates attack after the Kalman filter converges. Thus, from \eqref{eq:aprioristate} and \eqref{eq:Kalmanestimtion} we have:
	\begin{align}
	\hat{\boldsymbol{x}}(t) &= \left[\boldsymbol{I} -\tilde{\boldsymbol{K}} \boldsymbol{H} \right] \hat{\boldsymbol{x}}^{(-)}(t)+ \tilde{\boldsymbol{K}}\hspace{-2mm}\underbrace{\boldsymbol{z}(t)}_{\left[\boldsymbol{H}x(t)+\boldsymbol{e}(t)\right]}\hspace{-2mm}+\tilde{\boldsymbol{K}}\underbrace{\left[\begin{array}{c}
		\boldsymbol{a}_f\\\boldsymbol{a}_d
		\end{array}\right]}_{\tilde{\boldsymbol{a}}}=\left[\boldsymbol{I} -\tilde{\boldsymbol{K}} \boldsymbol{H} \right]\Big[\boldsymbol{A} \hat{\boldsymbol{x}}(t-1) + \boldsymbol{B}u_f(t-1) \nonumber\\
	&+ \boldsymbol{F} \hat{v}_l(t-1)\Big]+\tilde{\boldsymbol{K}}\boldsymbol{H}\left[\hat{\boldsymbol{x}}(t)+\boldsymbol{r}(t)\right]+\tilde{\boldsymbol{K}}\left[\boldsymbol{e}(t)+\tilde{\boldsymbol{a}}\right].
	\end{align}
	Thus, by simplifications analogous to the proof of Theorem \ref{Theorem:mudistribtution}, we can find $
	\boldsymbol{r}(t) = \boldsymbol{Q}\boldsymbol{A}\boldsymbol{r}(t-1)-\boldsymbol{Q}\tilde{\boldsymbol{K}}\tilde{\boldsymbol{a}}+\boldsymbol{\rho}(t)$. The expectation of $  \boldsymbol{r}(t) $ will be $
	\E\left\{\boldsymbol{r}(t)\right\} = \boldsymbol{Q}\boldsymbol{A}\E\left\{\boldsymbol{r}(t-1)\right\} - \boldsymbol{Q}\tilde{\boldsymbol{K}}\E\left\{\tilde{\boldsymbol{a}}\right\}.$ 
	Since the attack vector is constant, $ \tilde{\boldsymbol{a}} $, and $ \E\left\{\boldsymbol{r}(t)\right\} = \E\left\{\boldsymbol{r}(t-1)\right\} $ for $ t \rightarrow \infty $ we will have the steady state expected estimation error as $
	\E\left\{\boldsymbol{r}(t)\right\} \triangleq \bar{\boldsymbol{r}} = - \left[\boldsymbol{I} - \boldsymbol{Q}\boldsymbol{A}\right]^{-1} \boldsymbol{Q}\tilde{\boldsymbol{K}}\tilde{\boldsymbol{a}}.$
	Thus, considering such steady state expected estimation error, $ \bar{\boldsymbol{r}} $, we can derive the mean of $ \boldsymbol{\mu}(t) $:
	\begin{align}
	\E \left\{\boldsymbol{\mu}(t)\right\} = \E \left\{\boldsymbol{z}(t)-\boldsymbol{z}^{(-)}(t)\right\} = \E \left\{\boldsymbol{H}\boldsymbol{r}(t)+\tilde{\boldsymbol{a}}+\boldsymbol{e}(t)\right\} =\underbrace{\left[ \boldsymbol{I}  - \boldsymbol{H}\left[\boldsymbol{I} - \boldsymbol{Q}\boldsymbol{A}\right]^{-1} \boldsymbol{Q}\tilde{\boldsymbol{K}}\right]}_{\boldsymbol{\varPsi}}\tilde{\boldsymbol{a}}.\nonumber
	\end{align}
	Since $ \boldsymbol{a}_d $ is constant, then $ \boldsymbol{\sigma}_{\mu} $ will not change. Therefore, we can find the probability of not triggering the attack detection filter for $ \mu_i $ as in \eqref{eq:attacksuccess}.
\end{proof}
Proposition \ref{Proposition:probattack} derives the probability of staying stealthy for any particular DIA vector. Next we will find how much our defined attack detection filter is robust against a particular DIA. To this end, first we define a stealthy attack probability vector as $
	\boldsymbol{p}({\tilde{\boldsymbol{a}}})= \left[p_{1}({\tilde{\boldsymbol{a}}}),\dots, p_{n_f+n_d}({\tilde{\boldsymbol{a}}})\right]^T
$, which is a vector that shows the probability of staying stealthy by initiating $ \tilde{\boldsymbol{a}} $ on all of the $ v_f $ and $ d $ type sensors. The attacker has to find the optimal value of $ \tilde{\boldsymbol{a}} $ which maximizes the physical regret while staying stealthy with probability below a defined value, $ \boldsymbol{p} $, i.e., attack vector must be chosen such that the physical regret is maximized and every sensor $ i $ stays stealthy with the probability below the $ i $-th element of $ \boldsymbol{p} $.
\begin{corollary}\label{Theorem:MaximumErrorType1}
	 The attacker's optimal attack vector $ \boldsymbol{a}^* $ which maximizes the steady state physical regret by attacking $ v_f $ or $ d $ type sensors is the solution of following optimization problem:
	 \begin{align}\label{eq:theta}
	 &\argmax_{\tilde{\boldsymbol{a}}}\,\,\, \tilde{\theta}^2 \triangleq \left[\bar{\boldsymbol{r}}^T\boldsymbol{\varTheta}\bar{\boldsymbol{r}}+\boldsymbol{c}^T\bar{\boldsymbol{r}}\right]^2,\\
	  &\text{s.t.} \,\,\,\bar{\boldsymbol{r}} = - \left[\boldsymbol{I} - \boldsymbol{Q}\boldsymbol{A}\right]^{-1} \boldsymbol{Q}\tilde{\boldsymbol{K}}\tilde{\boldsymbol{a}},\label{eq:steady}\\
	  & \,\,\,\,\, \boldsymbol{\varTheta} \triangleq \left[\begin{array}{c c} 
	  \frac{1}{2_{bf}} & 0\\
	  0 & 0
	  \end{array}\right], \boldsymbol{c} = \left[\begin{array}{c}
	  \frac{\sqrt{\tilde{d}}}{b_f}+T \\
	  -1
	  \end{array}\right],
	  \\
	  &\boldsymbol{p}({\tilde{\boldsymbol{a}}}) \geq \boldsymbol{p},
	 \end{align}
	 where $ \tilde{d}$ is the average spacing between the ACVs.
\end{corollary}
\begin{proof}
From the definition of physical regret we have:
	\begin{align}
		R(t)\Big|_{u_f^*(t)} &= \left(\frac{\left(v_f(t)+Tu^*_f(t)\right)^2}{2b_f}-d(t)-Tv_l(t)+Tv_f(t)\right)^2\\ \nonumber
		& =\left(\underbrace{\frac{r_f^2(t)}{2b_f}+\frac{r_f(t)}{b_f}\sqrt{\underbrace{\hat{d}(t)+T\hat{v}_l(t)-T\hat{v}_f(t)}_{\simeq\tilde{d}}}-r_d(t)-Tr_l(t)+Tr_f(t)}_{\theta(t)}\right)^2,
	\end{align}
Moreover, the steady state value of $ \theta(t) $ can be derived by using the steady state expected value of $ \bar{\boldsymbol{r}} $ as $
		\bar{\theta} = \bar{\boldsymbol{r}}^T\boldsymbol{\varTheta}\bar{\boldsymbol{r}}+\boldsymbol{c}^T\bar{\boldsymbol{r}}.$
Also, the attacker must choose $ \boldsymbol{a}_d $ to satisfy \eqref{eq:attacksuccess}. Thus, the attacker must solve the optimization problems in \eqref{eq:theta}.
\end{proof}
Corollary \ref{Theorem:MaximumErrorType1} derives a robustness level for our proposed attack detection filter. It allows us to find the maximum regret which is caused by any stealthy attack vector. This allows us to design a secure cyber-physical ACV system since we take into account the physical regret of DIAs on sensors. However, the proposed attack detection filter works for sensors with a-priori information only. For ACV $ f $'s other sensors which measure $ v_l(t) $ we do not have a-priori information. Thus, we next address this case using MAB.

\subsection{Multi-armed bandit learning for attack detection in sensors without a-priori information}
To detect the attack on $ v_l $ type sensors, we cannot use any a priori estimation for $ v_l(t) $ as done for $ v_f $ and $ d $ type sensors because ACV $ f $ cannot have any information about evaluation of $ v_l(t) $ since it does not know how ACV $ l $ is being controlled. Such lack of a priori information makes the attack detection challenging. To overcome this challenge, we propose to use a MAB learning approach because such approach does not require a-priori information and finds the optimal action (detecting attacked sensors) only by interacting with the sensors and ACV's dynamics. Before applying the MAB algorithm, we define an a posteriori estimation for $ v_l(t) $:
\begin{align}\label{eq:vlapost}
	v_l^{(+)}(t) = \frac{\hat{d}(t+1)-\hat{d}(t)}{T} + \hat{v}_f(t).
\end{align}
Here, we also define an a posteriori residual as $ \left|\boldsymbol{\mu}_l^{(+)}(t)\right|\hspace{-1mm} \triangleq \hspace{-1mm} \left[|\mu_{{l}_1}^{(+)}(t)|,\dots, |\mu_{{l}_{n_l}}^{(+)}(t)|\right]^T\hspace{-1.5mm} \triangleq \hspace{-1mm} \Big|\boldsymbol{h}_lv_l^{(+)}(t)-\boldsymbol{z}_l(t)\Big| $. Next, we analyze the distribution of $ \boldsymbol{\mu}_l^{(+)}(t) $ to use it to design an attack detection filter.
\begin{theorem}\label{Theorem:aposteriori}
	$ \boldsymbol{\mu}_l^{(+)}(t) $ follows a Gaussian distribution with zero mean and covariance matrix $ \boldsymbol{C}_{\mu_l} $:
	\begin{align}
		\boldsymbol{C}_{\mu_l} =  \boldsymbol{\varUpsilon}\boldsymbol{R}_l\boldsymbol{\varUpsilon}^T+\boldsymbol{h}_l\boldsymbol{J}\left[ \boldsymbol{H}\boldsymbol{A}\boldsymbol{C}_r\boldsymbol{A}^T \boldsymbol{H}^T -\boldsymbol{R}\tilde{\boldsymbol{K}}^T\boldsymbol{Q}^T\boldsymbol{A}^T\boldsymbol{H}^T - \boldsymbol{H}\boldsymbol{A}\boldsymbol{Q}\tilde{\boldsymbol{K}}\boldsymbol{R}+\boldsymbol{R} \right]\boldsymbol{J}^T\boldsymbol{h}^T_l,
	\end{align} 
	where $ \boldsymbol{\varUpsilon} = \left[-\frac{s_1}{T} \boldsymbol{h}_l \boldsymbol{w}_l^T-\boldsymbol{I}\right]$ and $ \boldsymbol{J} = \left[\begin{array}{c c}
	0 & 1
	\end{array}\right]\tilde{\boldsymbol{K}} $.
\end{theorem}
\begin{proof}
See Appendix \ref{App:apost}.
\end{proof}
Theorem \ref{Theorem:aposteriori} derives the covariance matrix for $ \boldsymbol{\mu}_l^{(+)}(t) $ which we will use to detect anomaly in the $ v_l(t) $ type sensors. To this end, first, we define the \emph{squared Mahalanobis} (SM) distance, a measure which quantifies the distance between the a posteriori residual of sensors in a subset $ \mathcal{L} $ of $ v_l $ type sensors, $ {\boldsymbol{\mu}_l^{(+)}}(t,\mathcal{L}) $, and the distribution of a posteriori residuals which is derived in Theorem \ref{Theorem:aposteriori}. SM will help us find how much every subset $ \mathcal{L} $ deviates from its distribution.
\begin{definition}
	The \emph{squared Mahalanobis} distance between subset $ \mathcal{L} $'s a posteriori residual at time $ t $ and its distribution is defined as $
	D_{\mathcal{L}}(t) = {\boldsymbol{\mu}_l^{(+)}}^T(t,\mathcal{L})\boldsymbol{C}^{-1}_{\mu_l}(\mathcal{L}) {\boldsymbol{\mu}_l^{(+)}}(t,\mathcal{L}),$
	where $ D_{\mathcal{L}}(t) $ is the subset $ \mathcal{L} $'s SM and $ \boldsymbol{C}_{\mu_l}(\mathcal{L}) $ is the covariance matrix associated with $ \mathcal{L} $. 
\end{definition}
Next, we derive the expected value and variance of $ D_{\mathcal{L}}(t) $. We know from \cite[(378)]{matrixCookBook} that $ \E\left\{D_{\mathcal{L}}(t)\right\} = \E\left\{{\boldsymbol{\mu}_l^{(+)}}^T(t,\mathcal{L})\boldsymbol{C}^{-1}_{\mu_l}(\mathcal{L}) {\boldsymbol{\mu}_l^{(+)}}(t,\mathcal{L})\right\} = \text{Tr}\left(\boldsymbol{C}^{-1}_{\mu_l}(\mathcal{L})\boldsymbol{C}_{\mu_l}(\mathcal{L})\right) = \text{Tr}\left(\boldsymbol{I}\right) = |\mathcal{L}|, $ 	where $ |\mathcal{L}| $ is the number of sensors in $ \mathcal{L} $. In addition, for a subset $ \mathcal{L} $ of $ v_l $ type sensors and an attack vector $ \boldsymbol{a}_l(t) $, we will have:
 \begin{align}
 	\E\left\{D_{\mathcal{L}}(t)\right\} & = \E\left\{{\boldsymbol{\mu}_l^{(+)}}^T(t,\mathcal{L})\boldsymbol{C}^{-1}_{\mu_l}(\mathcal{L}) {\boldsymbol{\mu}_l^{(+)}}(t,\mathcal{L})+\boldsymbol{a}^T_l(t)\boldsymbol{C}^{-1}_{\mu_l}\boldsymbol{a}_l(t)\right\}\\ \nonumber
 	& = |\mathcal{L}|+\E\left\{\boldsymbol{a}^T_l(t,\mathcal{L})\boldsymbol{C}^{-1}_{\mu_l}(\mathcal{L})\boldsymbol{a}_l(t,\mathcal{L})\right\},
 \end{align}
 where $ \boldsymbol{a}_l(t,\mathcal{L}) $ consists of only elements of $ \boldsymbol{a}_l(t) $ which are in $ \mathcal{L} $. Since $ \boldsymbol{C}_{\mu}(\mathcal{L}) $ is positive semi definite then $ \boldsymbol{C}^{-1}_{\mu}(\mathcal{L}) $ is also positive semi definite. Thus, we will have $ \E\Big\{\boldsymbol{a}^T_l(t,\mathcal{L})\boldsymbol{C}^{-1}_{\mu_l}(\mathcal{L})\boldsymbol{a}_l(t,\allowbreak\mathcal{L}) \Big\} \geq 0 $.
 Therefore, if there exists a subset $  \mathcal{L} $ such that $ \E\left\{D_{\mathcal{L}}(t)\right\} \geq  |\mathcal{L}| $, then there is at least one sensor under attack in $ \mathcal{L} $. Hence, at each time step, to estimate $ \hat{v}_l(t) $, ACV $ f $ must choose a subset $ \mathcal{L} $ which has the least divergence from $ |\mathcal{L}| $. To capture the security level of a subset, we define a \emph{security divergence (SD)} metric for a subset $ \mathcal{L} $ as $ \varsigma_{\mathcal{L}}(t) \triangleq \frac{\E \left\{D_{\mathcal{L}}(t)\right\}}{|\mathcal{L}| }\geq 1$. Therefore, a subset $ \mathcal{L} $ with a lower $ \varsigma_{\mathcal{L}}(t) $ value has a higher security level.
 
 Moreover, any subset $ \mathcal{L} $ of $ v_l $ type sensors have a cost based on the estimation error induced to $ \hat{v}_l(t) $. Thus, we define the cost of $ \mathcal{L} $, as $ \nu_{\mathcal{L}} = \E\left\{\left(v_l(t) - \hat{v}_l(t,\mathcal{L}) \right)^2\right\}$, where $ \hat{v}_l(t,\mathcal{L}) $ is the estimation of $ v_l $ using the subset $ \mathcal{L} $. Thus, to find $ \nu_{\mathcal{L}} $ we have:
\begin{align}
	 \scriptstyle \E\left\{\left(v_l(t) - \hat{v}_l(t,\mathcal{L}) \right)^2\right\} =  \E\left\{\left(\sum_{i \in \mathcal{L}}w_{l_i}e_{l_i}\right)^2\right\} = \E\left\{\left( \frac{\sum_{i \in \mathcal{L}}\frac{1}{\sigma_{l_i}^2}e_{l_i}}{\sum_{i \in \mathcal{L}}\frac{1}{\sigma_{l_i}^2}}\right)^2\right\}=  \frac{ \sum_{i \in \mathcal{L}}\frac{1}{\sigma_{l_i}^4}\E\left\{e_{l^2_i}\right\}}{ \left(\sum_{i \in \mathcal{L}}\frac{1}{\sigma_{l_i}^2}\right)^2}\nonumber
	 =  \frac{ \sum_{i \in \mathcal{L}}\frac{1}{\sigma_{l_i}^2}}{ \left(\sum_{i \in \mathcal{L}}\frac{1}{\sigma_{l_i}^2}\right)^2} = \frac{1}{ \sum_{i \in \mathcal{L}}\frac{1}{\sigma_{l_i}^2}}.
\end{align}
Therefore, we have $ \nu_{\mathcal{L}} = \frac{1}{ \sum_{i\in \mathcal{L}}\frac{1}{\sigma_{l_i}^2}}$. Clearly, a subset $ \mathcal{L} $ with a higher $ \sum_{i\in \mathcal{L}}\frac{1}{\sigma_{l_i}^2} $ has a lower cost. 

While $ D_{\mathcal{L}} $ represents a security measure for a subset $ \mathcal{L} $, $ \nu_{\mathcal{L}} $ can be considered as an estimation performance metric. Thus, at each time step, ACV $ f $ must choose a subset $ \mathcal{L} $ which has the minimum SD as well as the minimum cost. \emph{Thus, we have to find the secure high performance subset $ \mathcal{L}^*(t) $ which is a subset that is secure and has the lowest estimation cost.} To this end, a minimization problem that finds the secure high performance subset can be defined as $
\mathcal{L}^*(t) = \argmin_{\mathcal{L}} \xi_{\mathcal{L}} (t), $ where, $ \xi_{\mathcal{L}} (t) = \nu_{\mathcal{L}}\varsigma_{\mathcal{L}}(t) $. Although $ \nu_{\mathcal{L}} $ is known for $ f $, $ \varsigma_{\mathcal{L}}(t) $ might change due to an attack on $ v_l $ type sensors. Thus, we propose an MAB learning algorithm which learns the secure high performance subset $ \mathcal{L}^*(t) $\cite{Auer2002,Kleinberg2010}. Our goal here is to choose a safe subset of sensors out of all available sensors by efficiently exploring different subsets and effectively exploiting the optimal subset, thus, we can apply an MAB framework to solve our problem.
 
 In an MAB problem, a decision maker (ACV $ f $), pulls an arm from a set of available arms (selects a subset $ \mathcal{L} $ from $ v_l $ type sensors). Each arm generates a cost after being
 played, based on a distribution that is not known to the decision maker. The aim of ACV $ f $ is to minimize a cumulative \emph{cyber} regret. This regret is defined as the difference between the reward of the best possible arm at each step, and the generated reward of the arm that is played \cite{Auer2002,Kleinberg2010,ali2018sleeping}.

 Let $  \xi^*(t) = \xi_{\mathcal{L}^*(t)}(t) $ be the lowest cost that could be achieved at time $ t $ from $ v_l $ type sensors. Thus, the cyber regret from $ t=t_0 $ up to a time $ t=t_e $ is defined as:
 \begin{align}
 	R^c(t_0,t_e) \triangleq \E \left[\sum_{t=t_0}^{t_e}\xi_{\mathcal{L}} (t)- \sum_{t=t_0}^{t_e}\xi^*(t)\right],
 \end{align}
 where the expectation is taken over the random choices of the $ \mathcal{L} $. Cyber regret implies choosing non-optimal and non-secure subset of sensors. Thus, having a high cyber regret implies choosing a subset of sensors with higher error and possibly higher injected data which can lead to an estimation error for $ \hat{v}_l(t) $. Since our proposed optimal controller which minimizes the physical regret is a function of estimated value, $ \hat{v}_l(t) $, thus, higher cyber regret results in a higher estimation error and can introduce physical regret to ACV $ f $. This clearly shows the cyber-physical aspects of ACV security.
 
 One of the most recognized methods for the MAB problem is to use the concept of UCB \cite{Auer2002}. In this method, the MAB algorithm at each time $t$ chooses $ \mathcal{L} $ such that:
 \begin{equation}\label{ucb}
 \mathcal{I} (t) = \argmax_{\mathcal{L}}  \frac{-\nu_{\mathcal{L}}}{|\mathcal{L}|}\frac{1}{n_{\mathcal{L}}} \sum_{i=1}^{n_{\mathcal{L}}} D_{\mathcal{L}}(t_{\mathcal{L},i}) + \sqrt{\frac{2 \ln t}{n_{\mathcal{L}}}}
 \end{equation}
 where $n_{\mathcal{L}}$ is the number of times that arm $ {\mathcal{L}} $ was played before and $ t_{\mathcal{L},i} $ is the time step when $ \mathcal{L} $ is selected for $ i $-th time. Also, note that $ n_{\mathcal{L}} \rightarrow \infty $, $ \frac{1}{n_{\mathcal{L}}} \sum_{i=1}^{n_{\mathcal{L}}} D_{\mathcal{L}}(t_{\mathcal{L},i}) \rightarrow \E\left\{D_{\mathcal{L}}(t)\right\} $\cite{Auer2002,Kleinberg2010}. The UCB algorithm has an expected cumulative regret of:
 \begin{align}
 \E \left\{R^c(t_0,t_e)\right\} = 8 \sum_{\mathcal{L}|\xi_{\mathcal{L}}<\xi^*} \frac{\ln (t_e-t_0+1)}{\E\left[\xi_{\mathcal{L}}\right]-\E\left[\xi^*\right]}+ \left(1+\frac{\pi^2}{3}\right)\sum_{\mathcal{L}|\xi_{\mathcal{L}}<\xi^*}\E\left[\xi_{\mathcal{L}}\right]-\E\left[\xi^*\right].
 \end{align}
 We use UCB algorithm because it has a logarithmic cyber regret and has been shown that there exists no algorithm that can have a better cyber regret \cite{Kleinberg2010}. Thus, using the UCB algorithm and the defined cyber regret we can address the DIA detection problem for sensors which lack a priori information.
\vspace{-8mm}
\section{Simulation Results and Analysis}\label{sec:sim}
For our simulations, we assume that ACV $ l $ has a sinusoidal speed pattern and ACV $ f $'s initial speed and spacing from the ACV $ l $ are $ v_f(0)=90 $ km/h and $ d(0)=100 $ m. We set $ n_f=n_d=n_l=4 $, $ T=0.1 $ s, $ b_f = 2.5  $ m/$ s^2 $, and $ u_f^{\text{max}}=-u_f^{\text{min}}=0.25 $ N/kg, which are chosen based on practical ACV characteristics\cite{Calandriello2011}.\vspace{-5mm}
\subsection{Optimal Safe Controller}
Fig. \ref{Fig:controller} shows an ACV which uses the proposed optimal controller. From Fig. \ref{Fig:controller}, we make three observations: 1) The physical regret converges to zero and the actual spacing converges to the optimal safe spacing after approximately 20 seconds which shows that the proposed controller works properly, 2) The speed of ACV $ f $, $ v_f(t) $, exhibits, approximately, a 5 seconds delay compared to the ACV $ l $'s speed vehicle $ v_l(t) $. This is because ACV $ f $ should first observe ACV $ l $ and then control its own speed, and 3) The estimated states match with the actual state values which shows that applied state estimation estimation processes have an error close to zero.

\begin{figure}
\begin{subfigure}{0.48\textwidth}
	\centering
	\includegraphics[width=\textwidth]{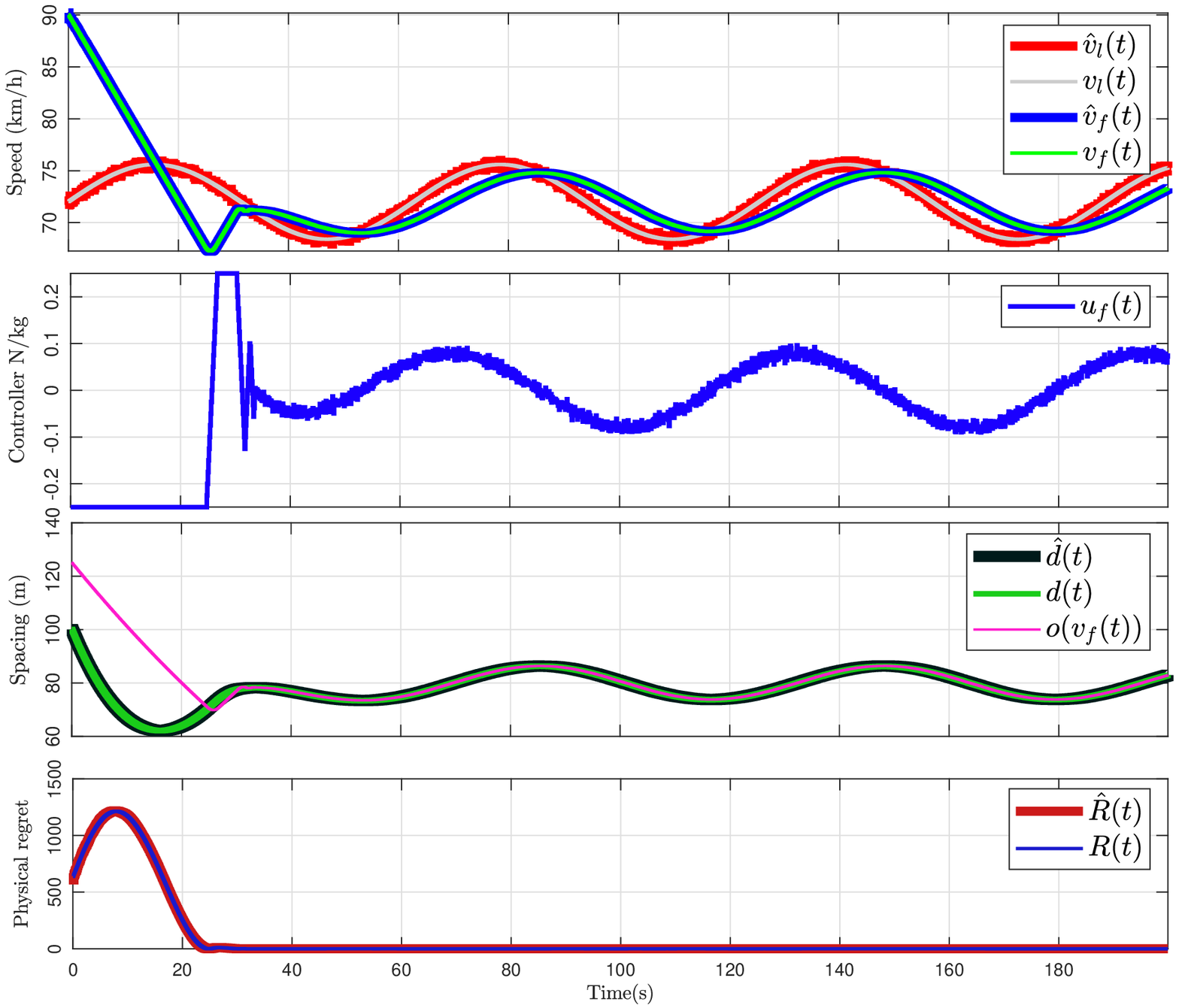}
	\vspace{-4mm}
	\caption{The proposed controller, physical regret, speed and spacing.}
	\label{Fig:controller}
	\vspace{-4mm}
\end{subfigure}~
\begin{subfigure}{0.48\textwidth}
\centering
\includegraphics[width=\textwidth]{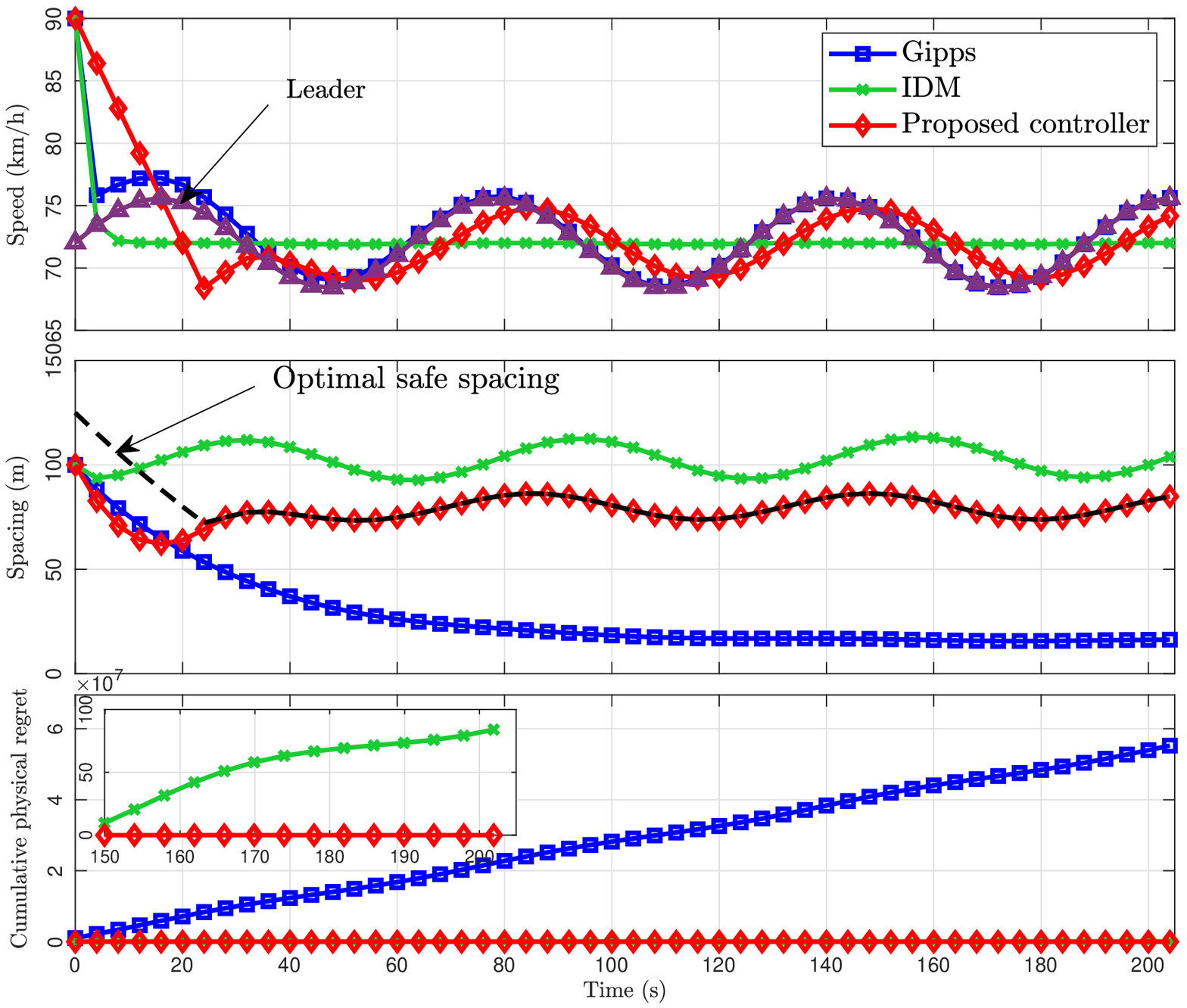}
\vspace{-4mm}
\caption{Comparison of proposed controller, Gipps, and IDM.}
\label{Fig:comp}
\vspace{-4mm}
\end{subfigure}
\caption{The proposed optimal controller's effect on the physical regret.}
\label{Fig:opt}
\vspace{-10mm}
\end{figure}

Fig. \ref{Fig:comp} shows a comparison between our proposed controller with Gipps \cite{gipps1981behavioural} and the intelligent driver model (IDM) \cite{Kesting4585} controllers, two of the well-known controllers for ACV leading-following scenarios. From Fig. \ref{Fig:comp}, we can observe three key points: 1) Our proposed controller can follow the speed of ACV $ l $ better compared to IDM as we can see that IDM converges to an almost constant speed while our proposed controller can track $ v_l(t) $, 2) The proposed controller converges to the optimal safe spacing. However, the IDM model has a positive offset from the optimal safe spacing and the Gipps controller always has a smaller spacing than the optimal safe spacing. This shows that the Gipps model does not take into consideration the safety as we can see from Fig. \ref{Fig:comp} that the spacing can admit small values close to 0 meters which increases the risk of collision. On the other hand, the IDM model does not consider the optimality of traffic flow as it always yields a higher spacing than the optimal spacing, and 3) The cumulative physical regret for our proposed optimal safe controller outperforms the other two models thus demonstrating that the proposed controller can jointly yield optimality and safety in ACV systems.\vspace{-5mm}
\subsection{Robustness Against Physical Attacks}
Next, we simulate an attacker which takes control of ACV $ l $ and after 100 seconds suddenly drops the speed of ACV $ l $ from 75 km/h to 5 km/h. Fig. \ref{Fig:physicalattack} shows the effect of this attack on the ACV $ f $'s speed, spacing and physical regret. Fig. \ref{Fig:physicalattack} shows that our proposed controller always maintains the optimal safe spacing even in presence of an attack. In contrast, the Gipps and IDM controllers always have an offset from the optimal safe spacing. Hence, our proposed controller is more robust against physical attacks. Note that, although Gipps can track the speed faster than our proposed method, however, it does not optimize the spacing. In addition, we can see from Fig. \ref{Fig:physicalattack} that the proposed controller has a physical regret closed to zero while the cumulative regret of IDM and Gipps grow linearly since they have a constant offset from the optimal safe spacing in this attack scenario.
\begin{figure}
	\centering
	\includegraphics[width=0.65\textwidth]{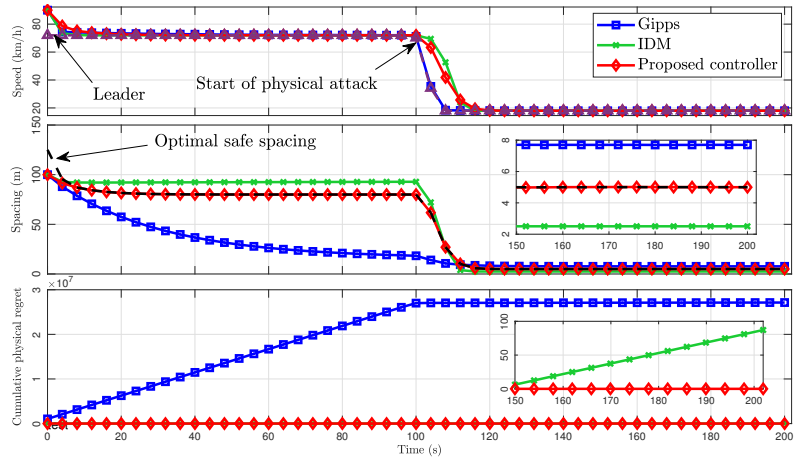}
	\vspace{-4mm}
	\caption{A comparison of the proposed controller, Gipps, and IDM when ACV system is physical attacked.}
	\label{Fig:physicalattack}
	\vspace{-11mm}
\end{figure}\vspace{-5mm}
\subsection{Cyber Security of ACV Systems}
To illustrate the effectiveness of our proposed attack detection approaches, in Fig. \ref{Fig:cyber}, we analyze the impact of DIA on physical regret and estimation errors. Fig. \ref{Fig:AttDom} shows the relationship between the stealthy attack probability and maximum DIA when we apply our proposed attack detection filter on the sensors with a priori information. Fig \ref{Fig:AttDom} shows that as the stealthy attack probability increases the domain of the attack reduces. Moreover the DIA has higher value for the sensors with higher noise variance. To stay stealthy with higher probability, the attacker must reduce the domain of its injected faulty data to the sensors. In addition, Figs. \ref{Fig:AttRegret}, \ref{Fig:Attvf}, and \ref{Fig:Attd} show how the physical regret and steady state errors are affected by the probability of staying stealthy. From these figures, we can see that as the stealthy attack probability increases the regret and the absolute value of steady state errors decrease because the attacker's maximum DIA decreases.

Fig. \ref{Fig:GreedyAttackDetection} shows the estimation error after applying our proposed attack detection filter on sensors with apriori estimation. From Fig. \ref{Fig:GreedyAttackDetection}, we can see that, while the attack on $ v_f(t) $ type sensor does not affect significantly the estimation, an attack on a $ d(t) $ type sensor can cause an estimation error on $ \hat{d}(t) $. In addition, Fig. \ref{Fig:GreedyAttackDetection} shows that the designed attack detection mechanism for sensors with a priori information mitigates DIAs and keeps the estimation error close to zero. 

\begin{figure}
	\begin{subfigure}{0.23\textwidth}
	\centering
	\includegraphics[width=\textwidth]{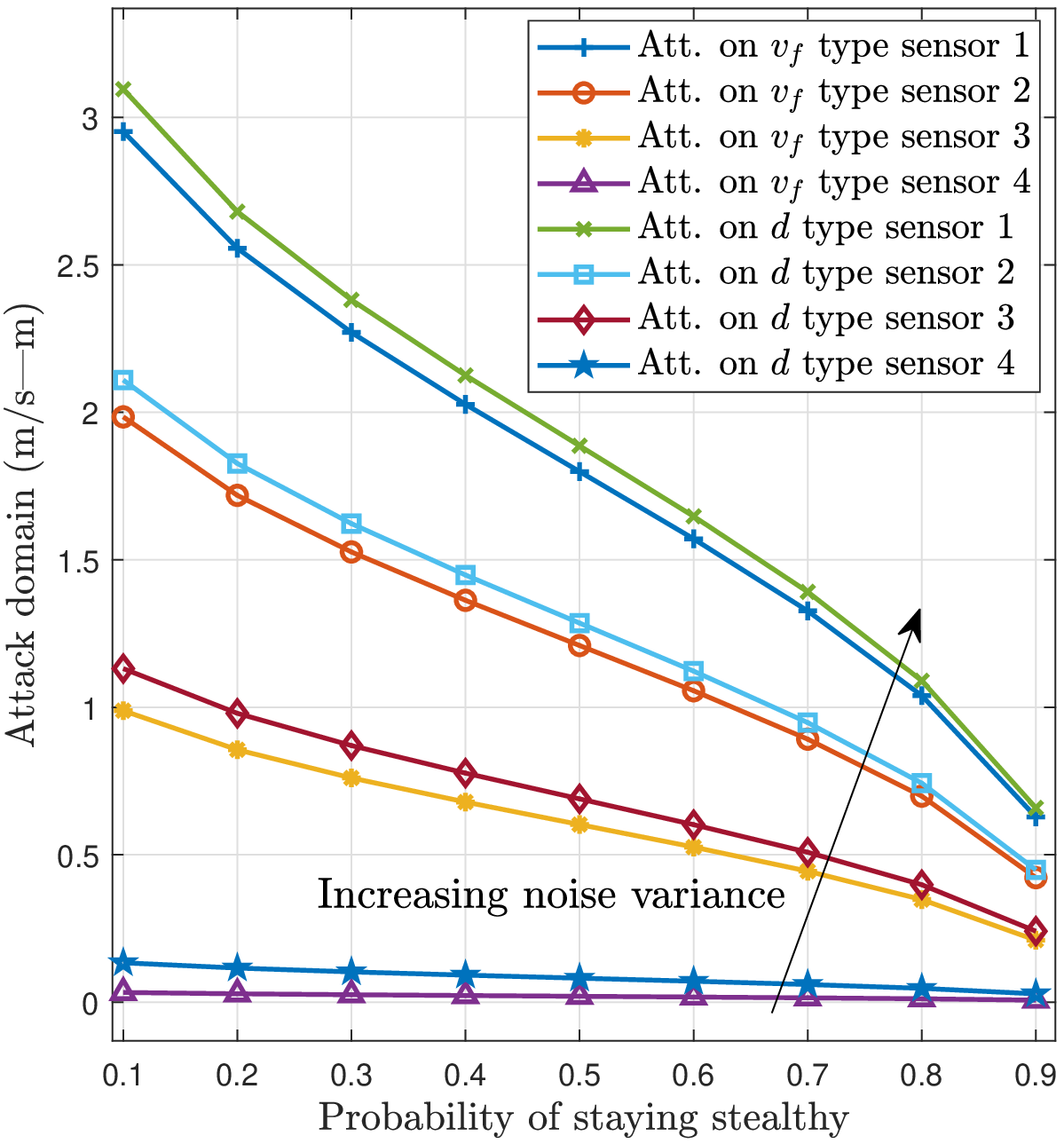}
	\caption{Attack detection filter applied on the cyber attacks on $ \hat{v}_f(t) $ and $ \hat{d}(t) $.}
	\label{Fig:AttDom}
	\vspace{-4mm}
\end{subfigure}~
\begin{subfigure}{0.23\textwidth}
	\centering
	\includegraphics[width=\textwidth]{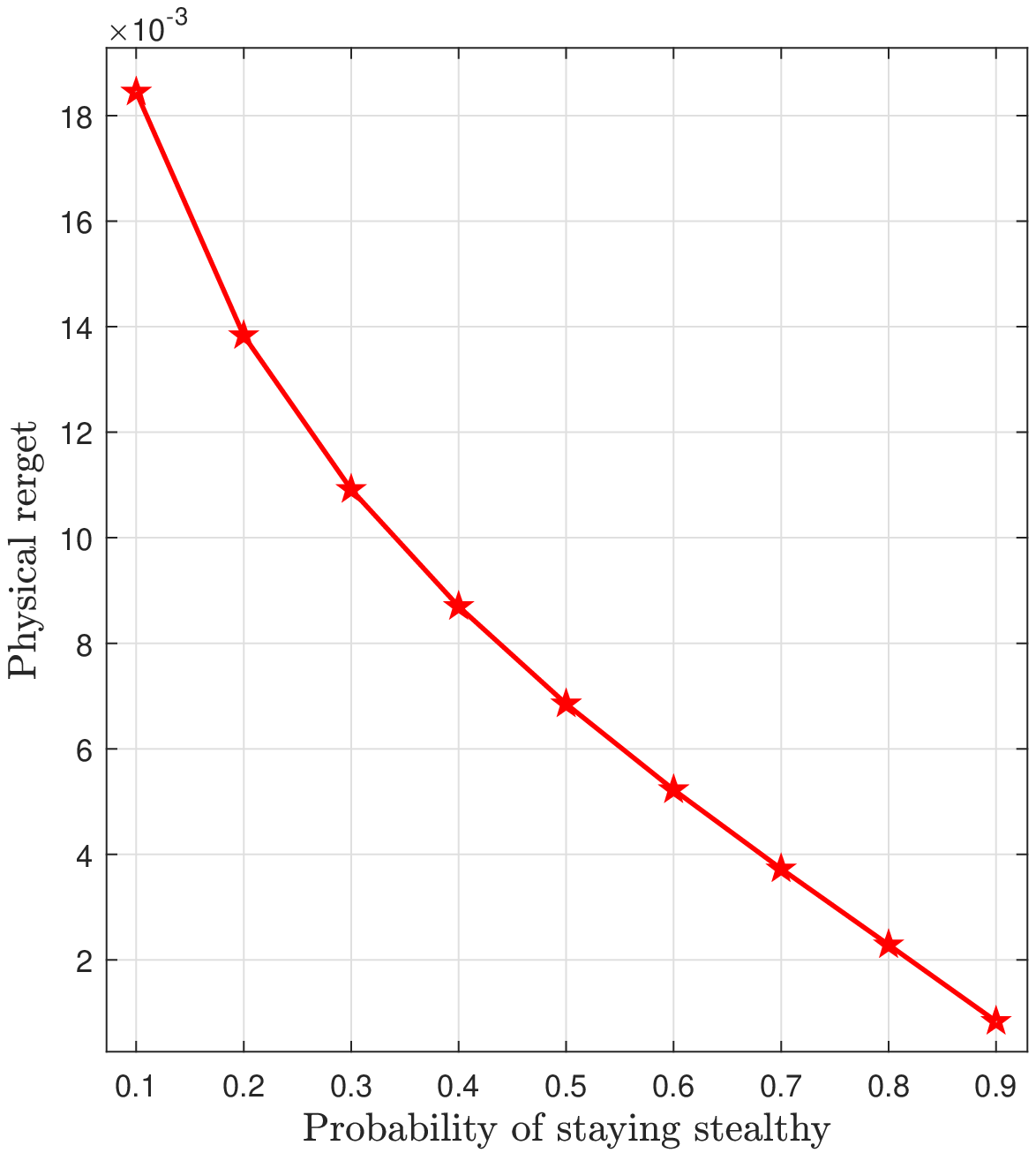}
	\caption{Attack detection filter applied on the cyber attacks on $ \hat{v}_f(t) $ and $ \hat{d}(t) $.}
	\label{Fig:AttRegret}
	\vspace{-4mm}
\end{subfigure}~
\begin{subfigure}{0.23\textwidth}
\centering
\includegraphics[width=\textwidth]{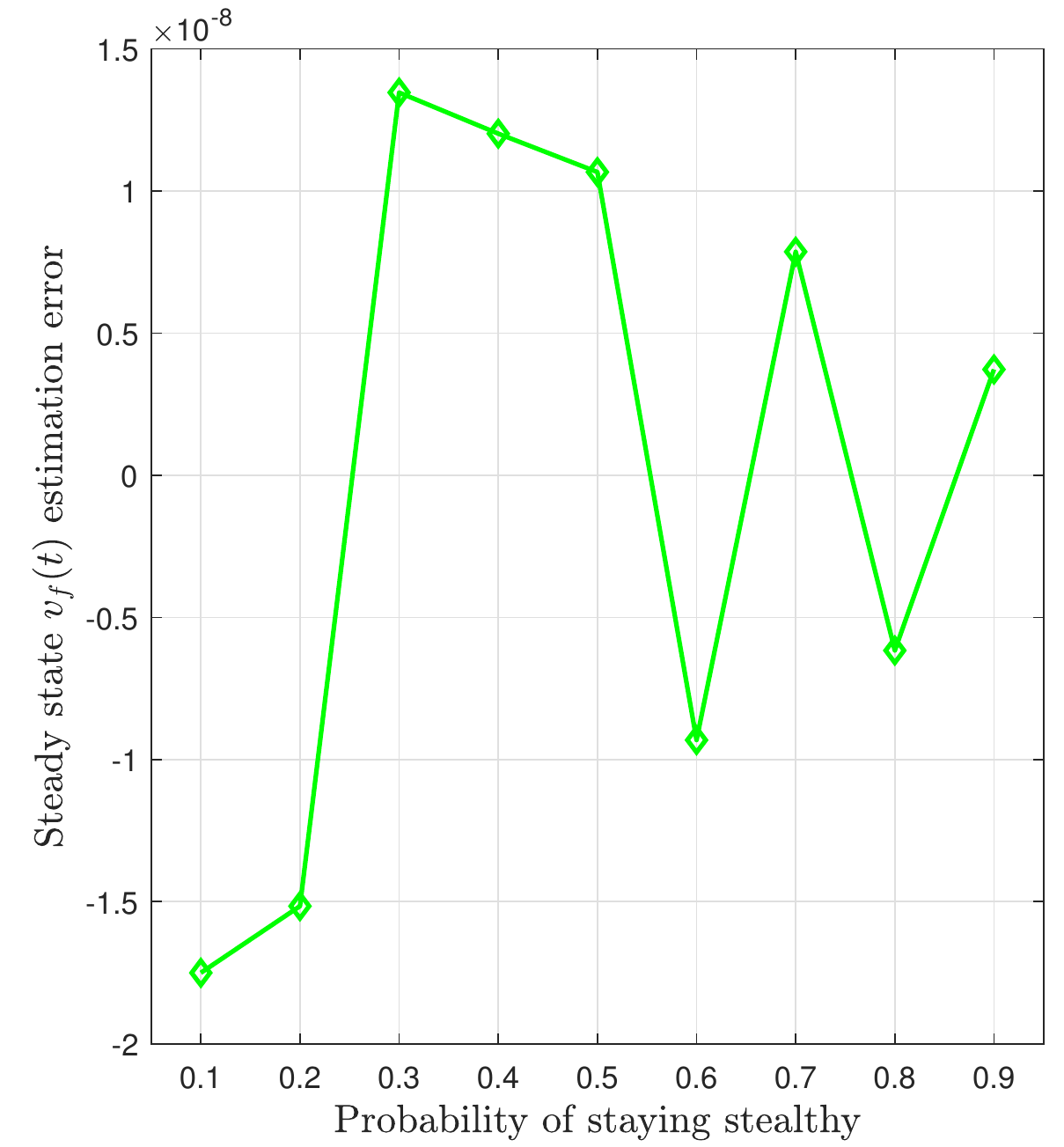}
\caption{Attack detection filter applied on the cyber attacks on $ \hat{v}_f(t) $ and $ \hat{d}(t) $.}
\label{Fig:Attvf}
\vspace{-4mm}
\end{subfigure}
~
\begin{subfigure}{0.23\textwidth}
	\centering
	\includegraphics[width=\textwidth]{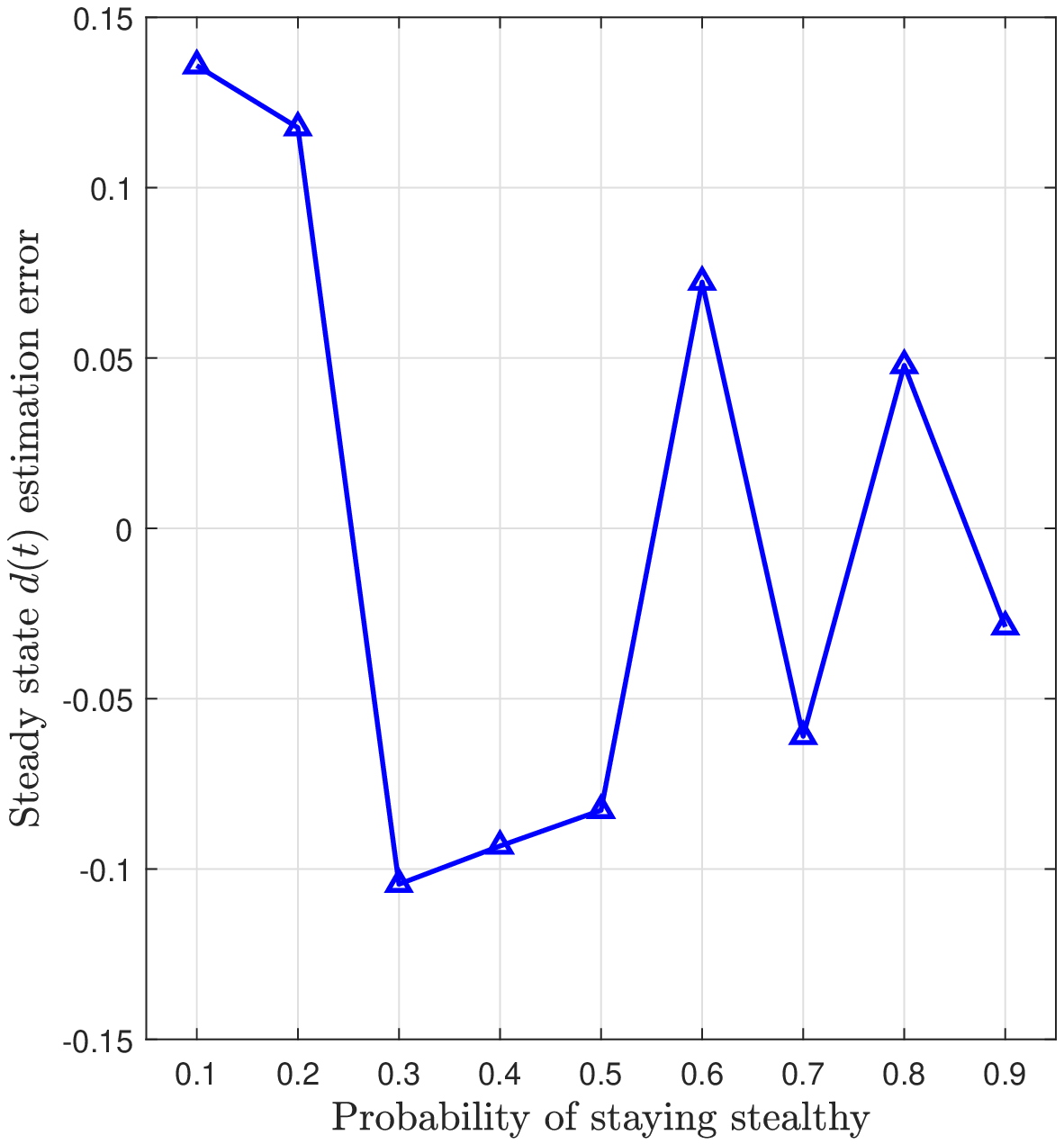}
	\caption{Attack detection filter applied on the cyber attacks on $ \hat{v}_f(t) $ and $ \hat{d}(t) $.}
	\label{Fig:Attd}
	\vspace{-4mm}
\end{subfigure}
\caption{Attack detection filter applied on the cyber attacks on $ \hat{v}_f(t) $ and $ \hat{d}(t) $.}
\label{Fig:cyber}
\vspace{-8mm}
\end{figure}

\begin{figure}[t]
	\centering
	\includegraphics[width=0.8\textwidth]{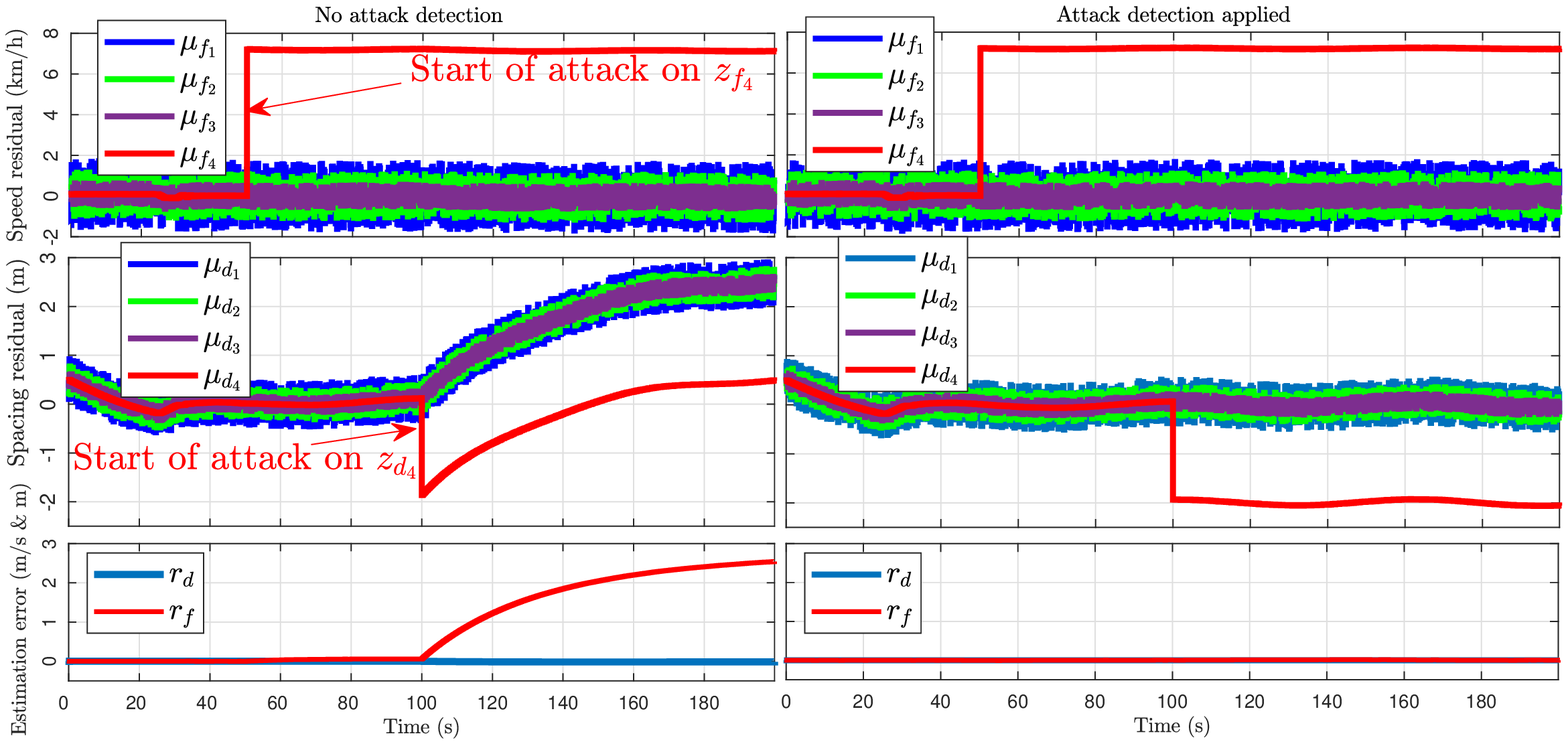}
	\vspace{-4mm}
	\caption{Attack detection filter applied on the cyber attacks on $ \hat{v}_f(t) $ and $ \hat{d}(t) $.}
	\label{Fig:GreedyAttackDetection}
	\vspace{-10mm}
\end{figure}

Fig. \ref{Fig:CumMul} shows how the increase on the number of under attack $ v_l $ type sensors affects the cyber regret. In this simulation, we use the same settings as for Fig. \ref{Fig:opt}. From Fig. \ref{Fig:CumMul}, we can first observe that irrespective of the number of under attack sensors, the attack can be detected in under 4 seconds as the cyber regret goes to zero after 4 second which is acceptable for this scenario as in 4 seconds ACV $ f $ travels for 80 meters while the optimal safe spacing is also 80 meters, thus, there will be no collision even if ACV $ l $ stops suddenly , while ACV $ f $'s sensors are attacked. Moreover, the growth rate of cyber regret for the three attacks is approximately equal regardless of the number of attacked sensors. Furthermore, Fig. \ref{Fig:CumMul} shows that, as the number of attacked sensors increases, the cumulative cyber regret increases faster thus the attack can be detected earlier. This shows that having access to more sensors is not essentially beneficial for the attacker and thus, the attacker must also optimize its set of attacked sensors to stay stealthy for a longer time.

\begin{figure}[t]
	\centering
	\includegraphics[width=0.75\textwidth]{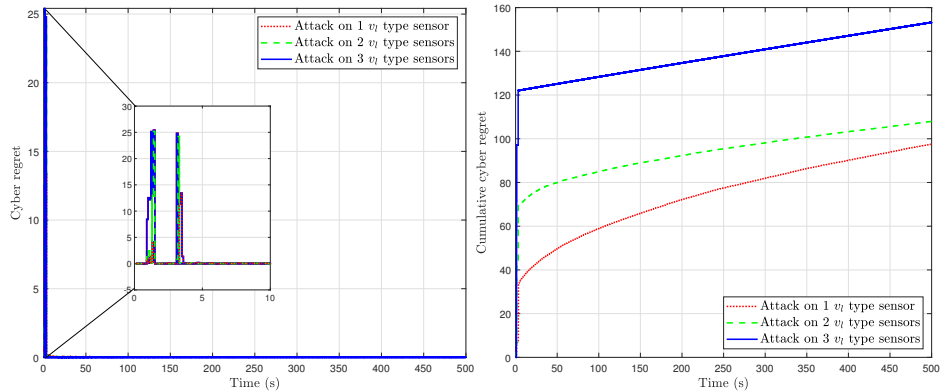}
	\vspace{-4mm}
	\caption{The cyber regret in MAB algorithm applied to $ v_l $ type sensors.}
	\label{Fig:CumMul}
	\vspace{-8mm}
\end{figure}

Fig. \ref{Fig:CumPow} shows how the domain of attack can affect the cyber regret. We can see from Fig. \ref{Fig:CumPow} that, as the domain of injected data to the $ v_l $ type sensor increases, the MAB algorithm can chooses the optimal subset more efficiently. This means that injecting higher values to the sensor increases the chance of detectability and thus, reduces the impact of the attack on the sensors. In addition, we can see that the cumulative cyber regret for the highest DIA (2.5 m/s data injection) is greater than other cases, at the beginning of the simulation. However, after almost 350 seconds, the cumulative cyber regret of the highest DIA becomes the least compared to the other DIAs. This is due to the fact that a larger DIA leads to a higher cyber regret at the beginning of the attack which leads the MAB algorithm to detect it earlier than other DIA cases.
\begin{figure}[t]
	\centering
	\includegraphics[width=0.75\textwidth]{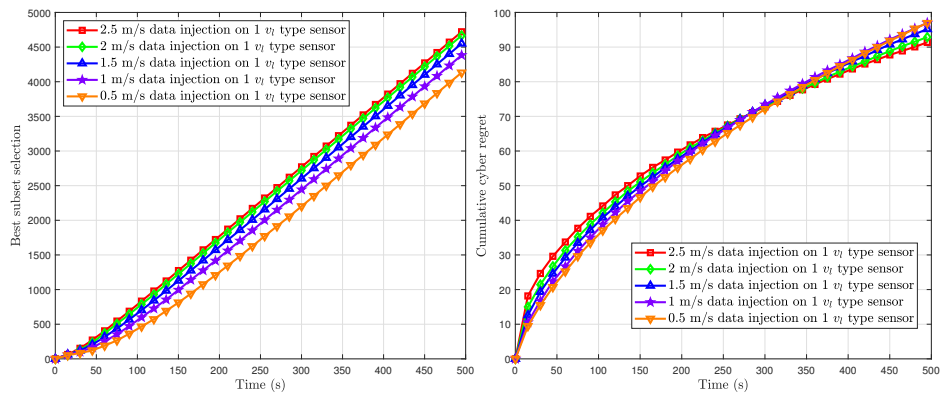}
	\vspace{-4mm}
	\caption{The best subset selection in MAB algorithm applied to $ v_l $ type sensors.}
	\label{Fig:CumPow}
	\vspace{-10mm}
\end{figure}

Fig. \ref{Fig:chosensubset} shows a scenario in which the attacker attacks only the third sensor when $ n_l = 4 $. Note that, $b_4b_3b_2b_1$ shows the state of the sensors such that for $  i=1,\dots,4 $ if $ b_i=1 $ then sensor $ i $ is under attack and otherwise, it is not. Obviously, when the attacker attacks the third sensor, the ACV must only rely on the other sensors. The $b_4b_3b_2b_1$ on the x-axis indicates the chosen subset by the MAB algorithm to use for estimation.  Fig. \ref{Fig:chosensubset} shows that using the MAB algorithm, percentage of the times that subset $ \{1011\} $ is chosen by MAB is higher than other cases, which means that the MAB algorithm can detect the attack on the $ v_l $ type sensors.
\begin{figure}[t]
	\centering
	\includegraphics[width=0.52\textwidth]{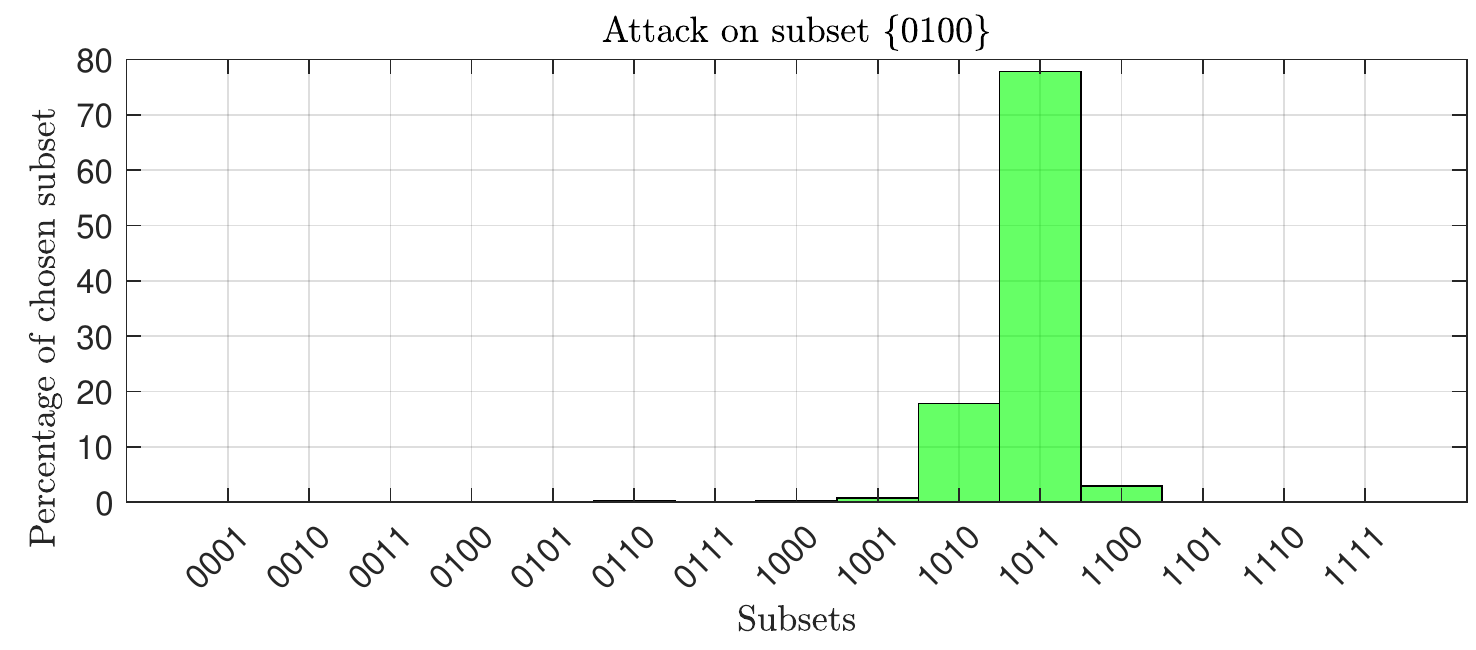}
	\vspace{-4mm}
	\caption{Attack detection with MAB algorithm applied to $ v_l $ type sensors.}
	\label{Fig:chosensubset}
	\vspace{-8mm}
\end{figure}
\begin{figure*}[t]
	\begin{subfigure}{0.32\textwidth}
		\includegraphics[width=\textwidth]{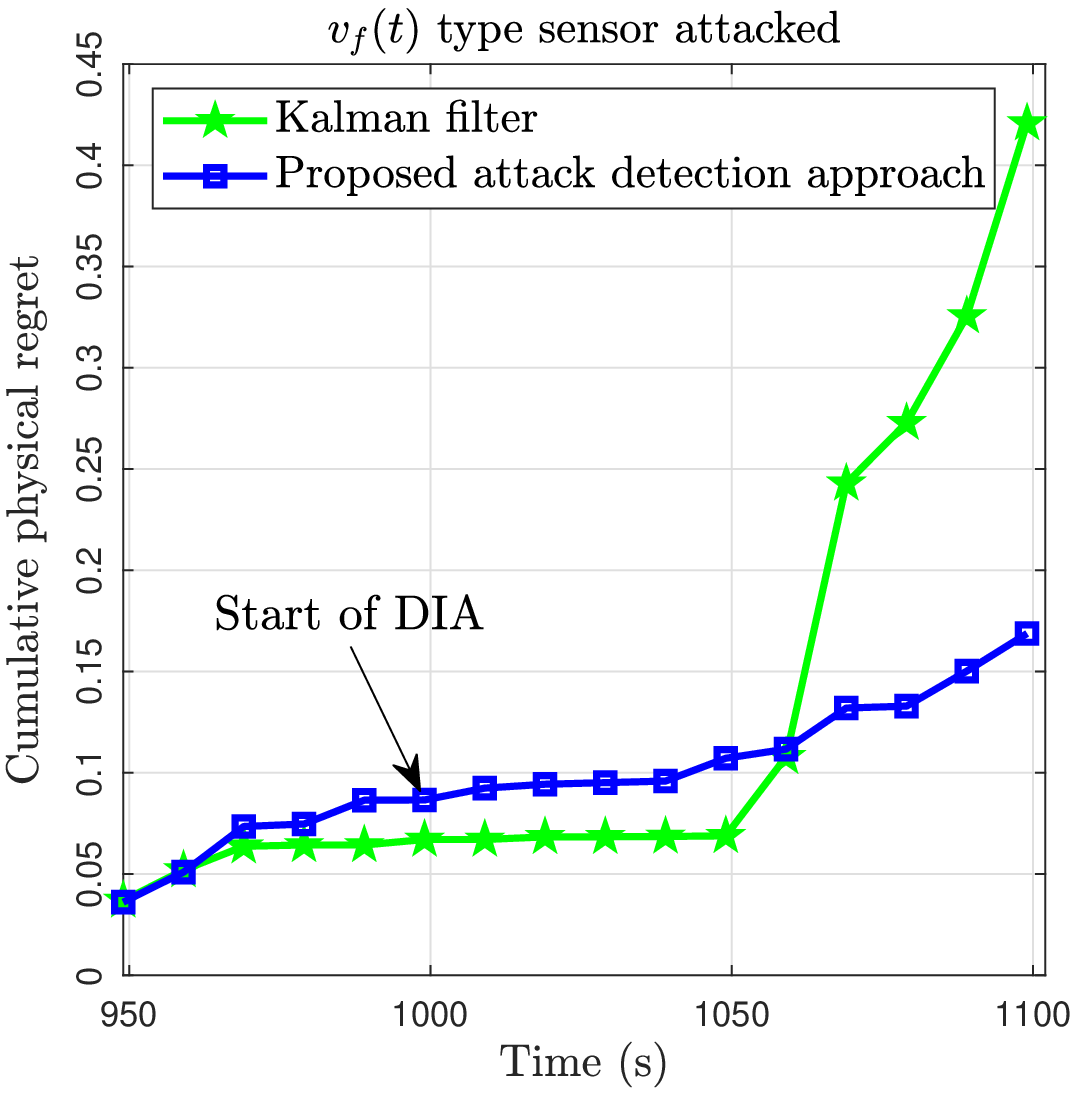}
		\caption{DIA on $ v_f(t) $ type.}
		\label{Fig:DIAvf}
		\vspace{-4mm}
	\end{subfigure}
~
\begin{subfigure}{0.32\textwidth}
	\includegraphics[width=\textwidth]{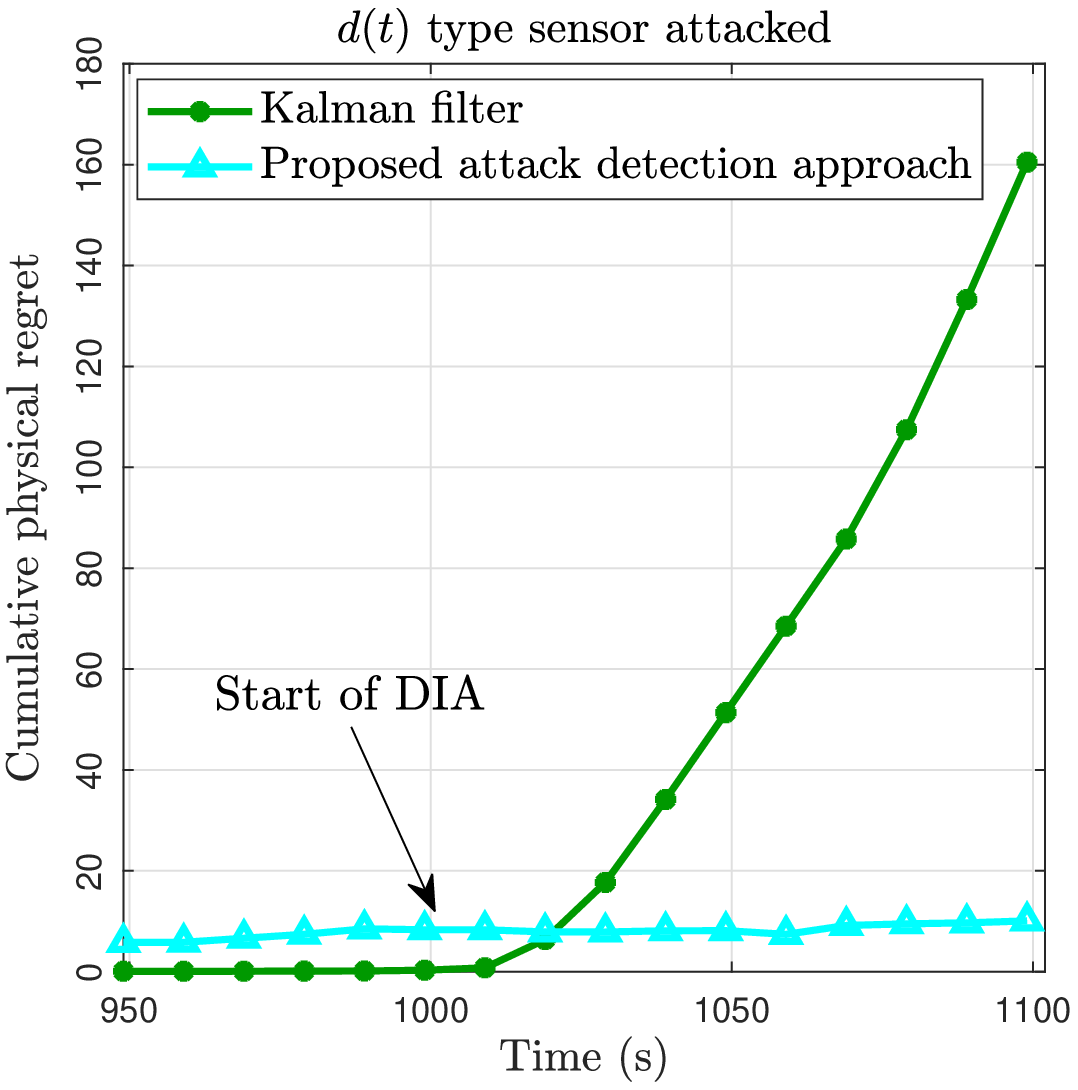}
	\caption{DIA on $ d(t) $ type.}
	\label{Fig:DIAd}
	\vspace{-4mm}
\end{subfigure}
~
\begin{subfigure}{0.32\textwidth}
	\includegraphics[width=\textwidth]{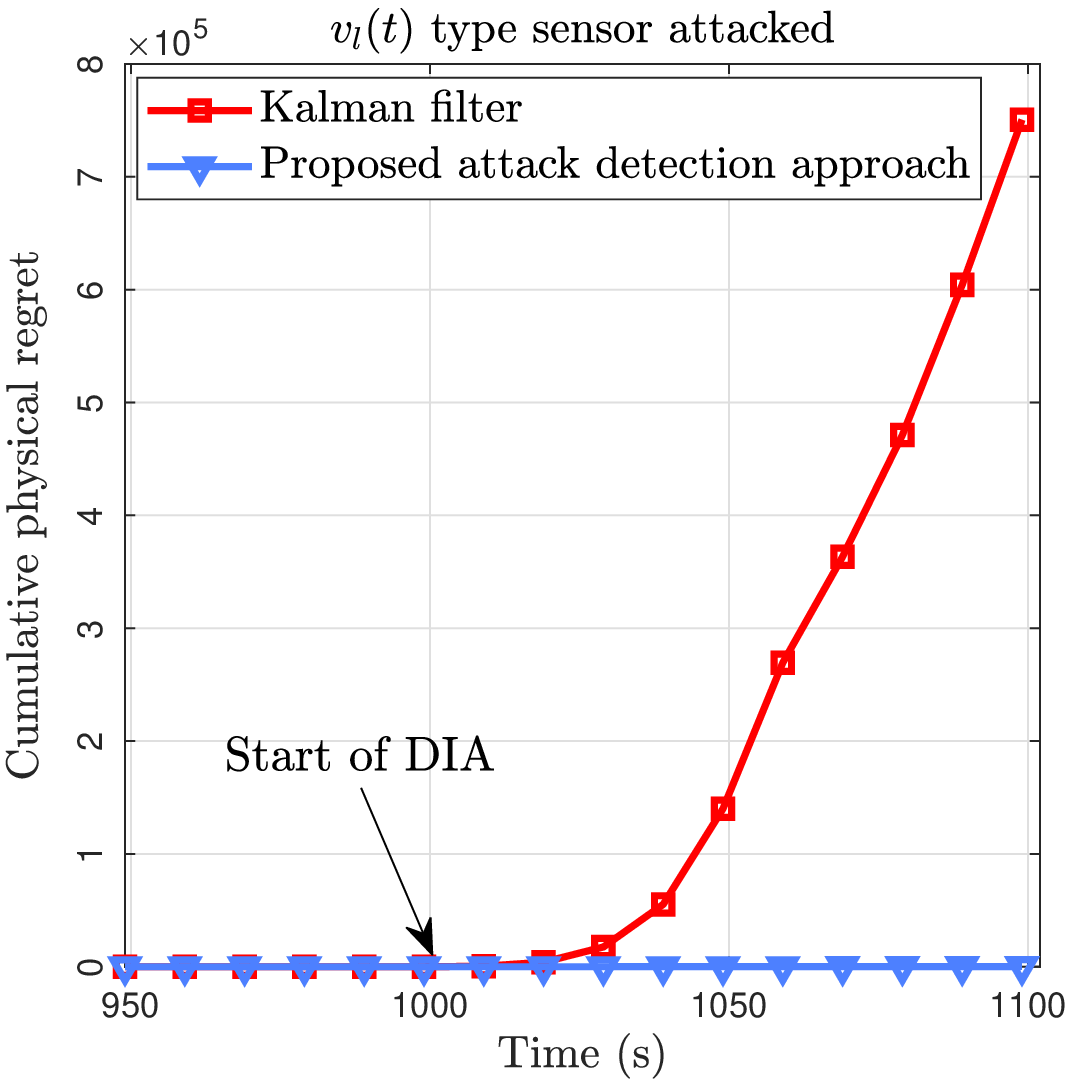}
	\caption{DIA on $ v_l(t) $ type.}
	\label{Fig:DIAvl}
	\vspace{-4mm}
\end{subfigure}
\caption{Comparison of the proposed DIA detection approaches with Kalman filtering.}
\label{Fig:DIA}
\vspace{-10mm}
\end{figure*}

Finally, to show that our proposed DIA detection approaches can reduce the physical regret significantly, in Fig. \ref{Fig:DIA}, we consider that the attacker attacks only one sensor of each type with the lowest noise error variance (most valuable sensor) after 1000 seconds (after ensuring that Kalman filter converged). Then, we compare the cumulative physical regret of the ACV system for a case in which the ACV only uses a Kalman filter for estimation and another case in which the ACV applies both of our proposed DIA detection filters. Moreover, in this attack, we consider that the attacker increases the sensor readings by a factor of $ 50\% $. As can be seen from Fig. \ref{Fig:DIA}, the proposed attack detection filter has lower physical regret compared to a simple Kalman filtering. We can see that while the physical regret increases linearly for Kalman filter, our proposed DIA detection approach keeps the physical regret close to zero. This shows that, our proposed DIA detection approach can successfully detect the attack, mitigate it, and keep the ACV system safe and secure, however the Kalman filter fails to stay robust against the DIA.\vspace{-5mm}
\section{Conclusion}\label{sec:conc}
In this paper, we have comprehensively studied the cyber-physical security and safety of ACV systems. We have proposed an optimal safe controller that maximizes the traffic flow and minimizes the risk of accidents on the roads by optimizing the speed and spacing of ACVs. In addition, the proposed optimal safe controller maximizes the stability of ACV systems and improves their robustness against physical attacks. Moreover, we have improved cyber-physical security of ACV systems by proposing two novel DIA detection approaches. The first approach detects DIAs using a priori information about the sensors. In the second approach, we have proposed an MAB algorithm to learn secure subset of sensors when there exists no a priori information about sensors. Simulation results have shown that the proposed optimal safe controller yields better stability and robustness compared to current state of the art controllers. In addition, the DIA detection approaches improve the security of the sensors and robustness of the ACV system compared to Kalman filtering.
\vspace{-3mm}
\begin{appendix}\vspace{-6mm}
	\subsection{Proof of Theorem \ref{Theorem1step}}\label{App:Theorem1}
	First, using \eqref{eq:spacingdisc} and \eqref{eq:1stepconst2} we can find that:
	\begin{align}\label{eq:constcomb}
	\max \hspace{-0.5mm}\left\{\hspace{-0.5mm}\frac{-v_f(t)}{T},\hspace{-0.5mm}u_f^{\text{min}},\hspace{-0.5mm}u_f(t-1)-\Delta u\hspace{-0.5mm}\right\}\leq u_f(t)\leq \min \left\{\frac{\tilde{v}-v_f(t)}{T},u_f^{\text{max}},u_f(t-1)+\Delta u \right\},
	\end{align}
	and since ACV $ f $ only has an estimation of $ v_f(t) $, then we can rewrite \eqref{eq:constcomb} as follows:
	\begin{align}
	\underbrace{\max\hspace{-0.5mm} \left\{\hspace{-0.5mm}\frac{-\hat{v}_f(t)}{T},u_f^{\text{min}},u_f(t-1)-\Delta u \hspace{-0.5mm}\right\}}_{\underbar{$u$}_1(t)}\leq u_f(t)\leq \underbrace{\min\hspace{-0.5mm} \left\{\hspace{-0.5mm}\frac{\tilde{v}-\hat{v}_f(t)}{T},u_f^{\text{max}},u_f(t-1)+\Delta u \hspace{-0.5mm}\right\}}_{\overline{u}_1(t)}.
	\end{align}
	Now, we can define Lagrangian multipliers and apply the Karush Kuhn Tucker (KKT) conditions to solve the optimization problem in \eqref{eq:1step}. Then, we will have:
	\begin{align}
	g(u_f(t),\lambda_1,\lambda_2) &= \left(\frac{(\hat{v}_f(t)+Tu_f(t))^2}{2b_f}-\hat{d}(t)-T\hat{v}_l(t)+T\hat{v}_f(t)\right)^2 + \lambda_1 (u_f(t)-\overline{u}_1(t)\nonumber\\
	&+\lambda_2(\underbar{$ u $}_1(t)-u_f(t)),
	\end{align}
	with the first order necessary conditions given by:
	$ 
	\frac{\partial g}{\partial u_f}=0,\,\, \frac{\partial g}{\partial \lambda_1}=0,$ and $ \frac{\partial g}{\partial \lambda_2}=0.
	$
	When $ \lambda_1 $ and $ \lambda_2 $ are not active ($ \lambda_1 =0$ and $ \lambda_2=0 $), we will have:
	\begin{align}
	\frac{\partial g}{\partial u_f}&= \frac{4T}{2b_f}\left(\frac{(\hat{v}_f(t)+Tu_f(t))^2}{2b_f}-\hat{d}(t)-T\hat{v}_l(t)+T\hat{v}_f(t)\right)\left(\hat{v}_f(t)+Tu_f(t)\right)=0\nonumber\\\label{eq:solutions}
	&\Rightarrow u^*_f(t) = \frac{1}{T}\left(\pm\sqrt{2b_f\left(\hat{d}(t)+T\hat{v}_l(t)-T\hat{v}_f(t)\right)}-\hat{v}_f(t)\right), \,\, -\frac{1}{T}\hat{v}_f(t).
	\end{align}
	However, $ \frac{1}{T}\left(-\sqrt{2b_f\left(\hat{d}(t)+T\hat{v}_l(t)-T\hat{v}_f(t)\right)}-\hat{v}_f(t)\right)$ and $ -\frac{1}{T}\hat{v}_f(t) $ are not acceptable solutions since they result in $ v_f(t+1)\leq 0 $. The next step is to show that the third solution in \eqref{eq:solutions} is the global minimum. Thus, we apply the second-order condition on $ \hat{R}(t) $ as follows:
	\begin{align}
	&\frac{d^2\hat{R}(t)}{du^2_f(t)}= \frac{4T}{2b_f}\left[\left(\frac{T(\hat{v}_f(t)+Tu_f(t))^2}{b_f}\right)
	+T\left(\frac{(\hat{v}_f(t)+Tu_f(t))^2}{2b_f}-\hat{d}(t)-T\hat{v}_l(t)+T\hat{v}_f(t)\right)\right]\nonumber\\
	&\frac{d^2\hat{R}(t)}{du^2_f(t)}\Bigg|_{u_f^*(t)}=\frac{4T}{2b_f}\left(\frac{T(\hat{v}_f(t)+Tu^*_f(t))^2}{b_f}\right)>0.
	\end{align}
	Thus, $ u_f^*(t) = \frac{1}{T}\left(\sqrt{2b_f\left(\hat{d}(t)+T\hat{v}_l(t)-T\hat{v}_f(t)\right)}-\hat{v}_f(t)\right) $ is a global maximizer of $ \hat{R}(t) $. Now, to consider the constraints in \eqref{eq:constcomb}, if we activate $ \lambda_1 $ or $ \lambda_2 $ the first order condition will result in having $ u^*_f(t)= \overline{u}_1(t) $ or $ \underbar{$ u $}_1(t) $. Therefore, the optimal 1-step ahead controller can be given by \eqref{eq:1stepopt}.
	\vspace{-7mm}\subsection{Proof of Theorem \ref{Theorem:mudistribtution}}\label{App:Theoremmu}
We know that the Kalman gain converges to $ \tilde{\boldsymbol{K}} $. Thus, from \eqref{eq:aprioristate} and \eqref{eq:Kalmanestimtion} we can find the state estimation $ \hat{\boldsymbol{x}}(t) $ as a process which depends on estimation error $ \boldsymbol{r}(t) $ as follows:
		\begin{align}
		\hat{\boldsymbol{x}}(t) &= \left[\boldsymbol{I} -\tilde{\boldsymbol{K}} \boldsymbol{H} \right] \hat{\boldsymbol{x}}^{(-)}(t)+ \tilde{\boldsymbol{K}}\underbrace{\boldsymbol{z}(t)}_{\left[\boldsymbol{H}x(t)+\boldsymbol{e}(t)\right]}\nonumber\\
		&=\left[\boldsymbol{I} -\tilde{\boldsymbol{K}} \boldsymbol{H} \right]\left[\boldsymbol{A} \hat{\boldsymbol{x}}(t-1) + \boldsymbol{B}u_f(t-1) + \boldsymbol{F} \hat{v}_l(t-1)\right]\nonumber\\
		&+\tilde{\boldsymbol{K}}\boldsymbol{H}\left[\hat{\boldsymbol{x}}(t)+\boldsymbol{r}(t)\right]+\tilde{\boldsymbol{K}}\boldsymbol{e}(t),\nonumber\\
		\Rightarrow \left[\boldsymbol{I} -\tilde{\boldsymbol{K}} \boldsymbol{H} \right]\hat{\boldsymbol{x}}(t) & = \left[\boldsymbol{I} -\tilde{\boldsymbol{K}} \boldsymbol{H} \right]\left[\boldsymbol{A} \hat{\boldsymbol{x}}(t-1) + \boldsymbol{B}u_f(t-1) + \boldsymbol{F} \hat{v}_l(t-1)\right]+\tilde{\boldsymbol{K}}\left[\boldsymbol{H}\boldsymbol{r}(t)+\boldsymbol{e}(t)\right],\nonumber\\
		\Rightarrow	\hspace{-0.2mm}\hat{\boldsymbol{x}}(t)\hspace{-0.5mm}  =\hspace{-0.5mm} \boldsymbol{A} \hat{\boldsymbol{x}}(t-1) \hspace{-0.5mm}&+ \hspace{-0.5mm}\boldsymbol{B}u_f(t-1)\hspace{-0.5mm} +\hspace{-0.5mm} \boldsymbol{F} \hat{v}_l(t-1)\hspace{-0.5mm}+\hspace{-0.5mm}\left[\hspace{-0.7mm}\boldsymbol{I} \hspace{-0.5mm}-\hspace{-0.5mm}\tilde{\boldsymbol{K}} \boldsymbol{H} \hspace{-0.5mm}\right]^{-1}\hspace{-1.2mm}\tilde{\boldsymbol{K}}\hspace{-0.5mm}\left[\hspace{-0.5mm}\boldsymbol{H}\boldsymbol{r}(t)\hspace{-0.5mm}+\hspace{-0.5mm}\boldsymbol{e}(t)\hspace{-0.7mm}\right]\hspace{-0.5mm}.\label{eq:estimated}
		\end{align}
		Now, by subtracting \eqref{eq:estimated} from the system model defined in \eqref{eq:system}, we will have:
		\begin{align}
		\boldsymbol{r}(t) &= \boldsymbol{A}\boldsymbol{r}(t-1)+\boldsymbol{F}{r}_l(t)-\left[\boldsymbol{I} -\tilde{\boldsymbol{K}} \boldsymbol{H} \right]^{-1}\tilde{\boldsymbol{K}}\left[\boldsymbol{H}\boldsymbol{r}(t)+\boldsymbol{e}(t)\right],\nonumber\\
		\boldsymbol{r}(t)&=\underbrace{\left[\boldsymbol{I}+\left[\boldsymbol{I} -\tilde{\boldsymbol{K}} \boldsymbol{H} \right]^{-1}\tilde{\boldsymbol{K}}\boldsymbol{H}\right]^{-1}}_{\boldsymbol{Q}}\hspace{-1mm}\boldsymbol{A}\boldsymbol{r}(t-1)	+\underbrace{\left[\boldsymbol{I}+\left[\boldsymbol{I} -\tilde{\boldsymbol{K}} \boldsymbol{H} \right]^{-1}\tilde{\boldsymbol{K}}\boldsymbol{H}\right]^{-1}\hspace{-2mm}\left[\boldsymbol{F}r_l(t)-\tilde{\boldsymbol{K}}\boldsymbol{e}(t)\right]}_{\boldsymbol{\rho}(t)}\nonumber
		\end{align}
		Moreover, we can see that $ \boldsymbol{\rho}(t) $ is a linear combination of independent white Gaussian process with zero mean. To find the covariance matrix $ \boldsymbol{C}_{\rho} $ of $ \boldsymbol{\rho}(t) $ we will have:		
		\begin{align}
		\boldsymbol{C}_{\rho} &= \E \left\{\boldsymbol{\rho}(t)\boldsymbol{\rho}^T(t)\right\} = \E \left\{\left[\boldsymbol{Q}\boldsymbol{F}r_l(t)-\boldsymbol{Q}\tilde{\boldsymbol{K}}\boldsymbol{e}(t)\right]\left[r_l^T(t)\boldsymbol{F}^T\boldsymbol{Q}^T-\boldsymbol{e}^T(t)\tilde{\boldsymbol{K}}^T\boldsymbol{Q}^T\right]\right\} \nonumber\\
		& = \E \Big\{\boldsymbol{Q}\boldsymbol{F}r_l(t)r_l^T(t)\boldsymbol{F}^T\boldsymbol{Q}^T-\boldsymbol{Q}\tilde{\boldsymbol{K}}\underbrace{\boldsymbol{e}(t)r_l^T(t)}_{0}\boldsymbol{F}^T\boldsymbol{Q}^T-\boldsymbol{Q}\boldsymbol{F}\underbrace{r_l(t)\boldsymbol{e}^T(t)}_{0}\tilde{\boldsymbol{K}}^T\boldsymbol{Q}^T\nonumber\\
		& +\boldsymbol{Q}\tilde{\boldsymbol{K}}\boldsymbol{e}(t) \boldsymbol{e}^T(t)\tilde{\boldsymbol{K}}^T\boldsymbol{Q}^T\Big\} = \boldsymbol{Q}\boldsymbol{F}\sigma^2_l\boldsymbol{F}^T\boldsymbol{Q}^T + \boldsymbol{Q}\tilde{\boldsymbol{K}}\boldsymbol{R}\tilde{\boldsymbol{K}}^T\boldsymbol{Q}^T.
		\end{align}
		Note that the fact that $  \boldsymbol{e}(t)$ and $r_l(t) $ are independent leads to having $ \boldsymbol{e}(t)r_l^T(t) = 0 $ .
		Next, we write the dynamic model for the estimation error as follows:
		\begin{align}\label{eq:rsteady}
		\boldsymbol{r}(t) = \boldsymbol{Q}\boldsymbol{A}\boldsymbol{r}(t-1)+\boldsymbol{\rho}(t)=\left[\boldsymbol{Q}\boldsymbol{A}\right]^{t-t_0}\boldsymbol{r}(t_c) + \sum_{k=0}^{t-t_0-1}\left[\boldsymbol{Q}\boldsymbol{A}\right]^{t-t_0-k-1}\boldsymbol{\rho}(t_c+k),
		\end{align}
		where $ t_0 $ is the time step where the Kalman filter converges.
		Thus, for $ t\gg t_0 $ and if $ \boldsymbol{Q}\boldsymbol{A} $ is asymptotically stable then we will have $ \boldsymbol{r}(t) $ as a Gaussian process due to the central limit theorem. Now, to find the mean of $ \boldsymbol{r}(t) $ we can write:
		\begin{align}
		\E\left\{\boldsymbol{r}(t)\right\}=\boldsymbol{Q}\boldsymbol{A}\E\left\{\boldsymbol{r}(t-1)\right\}+\E\left\{\boldsymbol{\rho}(t)\right\} =  \boldsymbol{Q}\boldsymbol{A}\E\left\{\boldsymbol{r}(t)\right\}, \Rightarrow \E\left\{\boldsymbol{r}(t)\right\} = \boldsymbol{0}.
		\end{align}
		Next, since the mean is zero, then the covariance matrix of  $ \boldsymbol{r}(t) $, $ \boldsymbol{C}_r = \E \{\boldsymbol{r}(t)\boldsymbol{r}^T(t)\} $. To find $ \boldsymbol{C}_r$, we have:
		\small
		\begin{align}
		{\textstyle\E \left\{\boldsymbol{r}(t)\boldsymbol{r}^T(t)\right\}} &{\textstyle= \E\left\{\left[\boldsymbol{Q}\boldsymbol{A}\boldsymbol{r}(t-1) + \boldsymbol{\rho}(t)\right]\left[\boldsymbol{r}(t-1)\boldsymbol{A}^T\boldsymbol{Q}^T + \boldsymbol{\rho}^T(t)\right]\right\}}
		\nonumber\\ 
		&{\textstyle=\E \Big\{\boldsymbol{Q}\boldsymbol{A}\boldsymbol{r}(t-1)\boldsymbol{r}^T(t-1)\boldsymbol{A}^T\boldsymbol{Q}^T+ \boldsymbol{Q}\boldsymbol{A}\underbrace{\boldsymbol{r}(t-1)\boldsymbol{\rho}^T(t)}_{0}+\underbrace{\boldsymbol{\rho}(t)\boldsymbol{r}^T(t-1)}_{0}\boldsymbol{A}^T\boldsymbol{Q}^T} + \boldsymbol{\rho}(t)\boldsymbol{\rho}^T(t)\Big\}.\nonumber
		\end{align}
		\normalsize
		Note that since $ \boldsymbol{r}(t-1)$ and $\boldsymbol{\rho}(t) $ are independent $ \boldsymbol{\rho}(t)\boldsymbol{r}^T(t-1)=0 $. In addition, we know that after convergence we will have $ \E \left\{\boldsymbol{r}(t)\boldsymbol{r}^T(t)\right\} = \E \left\{\boldsymbol{r}(t-1)\boldsymbol{r}^T(t-1)\right\} = \boldsymbol{C}_{r} $ then, we will need to solve the following \emph{discrete Ricatti} equation $
		\boldsymbol{C}_{r} = \boldsymbol{Q}\boldsymbol{A}\boldsymbol{C}_r\boldsymbol{A}^T\boldsymbol{Q}^T + \boldsymbol{C}_{\rho}.$
		
		Now, the final step is to find the covariance matrix of the residual $ \boldsymbol{C}_{\mu} $. Since after Kalman filter convergence we will have $ \hat{\boldsymbol{x}}^{(-)}(t) = \hat{\boldsymbol{x}}(t) $ then:
		\begin{align}
		\boldsymbol{z}(t)-\boldsymbol{z}^{(-)}(t)&=\boldsymbol{H}\boldsymbol{x}(t)+\boldsymbol{e}(t)-\boldsymbol{H}\boldsymbol{x}^{(-)}(t)=\boldsymbol{H}\boldsymbol{x}(t)+\boldsymbol{e}(t)-\boldsymbol{H}\hat{\boldsymbol{x}}(t)=\boldsymbol{H}\boldsymbol{r}(t)+\boldsymbol{e}(t).
		\end{align}
		Thus, we will have:
		\begin{align}
		\E\left\{\boldsymbol{\mu}(t)\boldsymbol{\mu}^T(t)\right\} &= \E \left\{\left[\boldsymbol{H}\boldsymbol{r}(t)+\boldsymbol{e}(t)\right]\left[\boldsymbol{r}^T(t)\boldsymbol{H}^T+\boldsymbol{e}^T(t)\right]\right\}\nonumber\\
		& = \E \left\{\boldsymbol{H}\boldsymbol{r}(t)\boldsymbol{r}^T(t)\boldsymbol{H}^T+\boldsymbol{e}(t)\boldsymbol{r}^T(t)\boldsymbol{H}^T+\boldsymbol{H}\boldsymbol{r}(t)\boldsymbol{e}^T(t)+\boldsymbol{e}(t)\boldsymbol{e}^T(t)\right\}\nonumber\\
		&= \boldsymbol{H}\boldsymbol{C}_r\boldsymbol{H}^T + \boldsymbol{R} + \E\left\{\boldsymbol{e}(t)\boldsymbol{r}^T(t)\boldsymbol{H}^T+\boldsymbol{H}\boldsymbol{r}(t)\boldsymbol{e}^T(t)\right\}.
		\end{align}
		Now, we have:\small
		\begin{align}
		\E\left\{\boldsymbol{e}(t)\boldsymbol{r}^T(t)\boldsymbol{H}^T\right\} & = 	\E\left\{\underbrace{\boldsymbol{e}(t)\boldsymbol{r}^T(t-1)}_{0}\boldsymbol{A}^T\boldsymbol{Q}^T\boldsymbol{H}^T+\underbrace{\boldsymbol{e}(t)\Big(r_l(t)}_{0}{F}^T-\boldsymbol{e}^T(t)\tilde{\boldsymbol{K}}^T\Big)\boldsymbol{Q}^T\boldsymbol{H}^T\right\} \nonumber= -\boldsymbol{R}\tilde{\boldsymbol{K}}^T\boldsymbol{Q}^T\boldsymbol{H}^T,\\
		\E\left\{\boldsymbol{H}\boldsymbol{r}(t)\boldsymbol{e}^T(t)\right\} & = 	\E\left\{\boldsymbol{H}\boldsymbol{Q}\boldsymbol{A}\underbrace{\boldsymbol{r}(t-1)\boldsymbol{e}^T(t)}_{0}+\boldsymbol{H}\boldsymbol{Q}\Big(-\tilde{\boldsymbol{K}}\boldsymbol{e}(t)\underbrace{+{\boldsymbol{F}r_l(t)}\Big)\boldsymbol{e}^T(t)}_{0}\right\}\nonumber = - \boldsymbol{H}\boldsymbol{Q}\tilde{\boldsymbol{K}}\boldsymbol{R}.	
		\end{align}\normalsize
		Thus, the covariance matrix can be given by \eqref{eq:covmu}.\vspace{-6mm}
\subsection{Proof of Theorem \ref{Theorem:aposteriori}}\label{App:apost}
From \eqref{eq:Kalmanestimtion}, after Kalman filter convergence we obtain:
\begin{align}
\hat{\boldsymbol{x}}(t+1) =& \left[\boldsymbol{I}-\tilde{\boldsymbol{K}}\boldsymbol{H}\right]\boldsymbol{x}^{(-)}(t) + \tilde{\boldsymbol{K}}\left[\boldsymbol{H}\boldsymbol{x}(t+1)+\boldsymbol{e}(t)\right]\nonumber\\
= & \left[\boldsymbol{I}-\tilde{\boldsymbol{K}}\boldsymbol{H}\right]\left[\boldsymbol{A}\hat{\boldsymbol{x}}(t)+\boldsymbol{B}{u}_f(t)+\boldsymbol{F}\hat{{v}}_l(t)\right]+ \tilde{\boldsymbol{K}}\left[\boldsymbol{H}\left[\boldsymbol{A}{\boldsymbol{x}}(t)+\boldsymbol{B}{u}_f(t)+\boldsymbol{F}{{v}}_l(t)\right]+\boldsymbol{e}(t)\right]\nonumber\\
= & \boldsymbol{A}\hat{\boldsymbol{x}}(t)+\boldsymbol{B}u_f(t)+\left[\boldsymbol{I}-\tilde{\boldsymbol{K}}\boldsymbol{H}\right]\boldsymbol{F}\hat{v}_l(t) + \tilde{\boldsymbol{K}} \boldsymbol{H} \boldsymbol{F} v_l(t) + \tilde{\boldsymbol{K}} \left[\boldsymbol{H}\boldsymbol{A}\boldsymbol{r}(t)+\boldsymbol{e} (t) \right].\label{eq:apost}
\end{align}
Due to the definition of\small $ \boldsymbol{F} = \left[\begin{array}{c c}
0 & T
\end{array}\right]^T $, $ \left[\boldsymbol{I}-\tilde{\boldsymbol{K}}\boldsymbol{H}\right]\boldsymbol{F} $\normalsize~and \small$ \tilde{\boldsymbol{K}}\boldsymbol{H}\boldsymbol{F} $\normalsize~are vectors with two elements where the first elements are zero. We define two parameters $ s_1 $ and $ s_2 $ as \small
$
\left[\boldsymbol{I}-\tilde{\boldsymbol{K}}\boldsymbol{H}\right]\boldsymbol{F} = \left[\begin{array}{c}
0\\
s_1
\end{array}\right]$ and $ \tilde{\boldsymbol{K}}\boldsymbol{H}\boldsymbol{F} = \left[\begin{array}{c}
0\\
s_2
\end{array}\right],
$ \normalsize
and we will have $ s_1+ s_2 =T $. Then using \eqref{eq:vlapost} and \eqref{eq:apost} the a posteriori estimation will be \small$
v_l^{(+)}(t) = \frac{s_1}{T}\hat{v}_l(t)+ \frac{s_2}{T}v_l(t) + \underbrace{\left[\begin{array}{c c}
	0 & 1
	\end{array}\right]\tilde{\boldsymbol{K}}}_{\boldsymbol{J}}\left[ \boldsymbol{H}\boldsymbol{A}\boldsymbol{r}(t)+\boldsymbol{e} (t) \right]. $\normalsize
Thus, we will have:
\begin{align}
\boldsymbol{\mu}_l^{(+)}(t)&= \boldsymbol{h}_lv_l^{(+)}(t)-\boldsymbol{h}_lv_l(t)-\boldsymbol{e}_l(t)\nonumber 
=\boldsymbol{h}_l\left[\frac{s_1}{T}\left(\hat{v}_l(t)-v_l(t)\right)\right]-\boldsymbol{e}_l(t)+\boldsymbol{h}_l\boldsymbol{J}\left[ \boldsymbol{H}\boldsymbol{A}\boldsymbol{r}(t)+\boldsymbol{e} (t) \right]\\
&=\underbrace{-\boldsymbol{h}_l\left[\frac{s_1}{T}\sum_{i=1}^{n_l}w_{l_i}e_{l_i}\right]-\boldsymbol{e}_l(t)}_{\boldsymbol{\varUpsilon}\boldsymbol{e}_l(t)}+\boldsymbol{h}_l\boldsymbol{J}\left[ \boldsymbol{H}\boldsymbol{A}\boldsymbol{r}(t)+\boldsymbol{e} (t) \right].
\end{align}
First, we derive the mean of a posteriori residual:
\begin{align}
\E\hspace{-0.5mm} \left\{\hspace{-0.5mm}\boldsymbol{\mu}_l^{(+)}(t)\hspace{-0.5mm}\right\} = -\boldsymbol{h}_l\hspace{-0.5mm}\left[\hspace{-0.5mm}\frac{s_1}{T}\sum_{i=1}^{n_l}w_{l_i}\E\left\{\hspace{-0.5mm}e_{l_i}\hspace{-0.5mm}\right\}\hspace{-0.5mm}\right]-\E \left\{\boldsymbol{e}_l(t)\right\}+\boldsymbol{h}_l\boldsymbol{J}\left[ \boldsymbol{H}\boldsymbol{A}\E \left\{\boldsymbol{r}(t)\right\}+\E \left\{\boldsymbol{e} (t)\right\} \right] = \boldsymbol{0}.
\end{align}
Since the mean is zero the covariance matrix of $ \boldsymbol{\mu}_l^{(+)}(t) $, $ \boldsymbol{C}_{\mu_l} $, can be derived as follows:
\begin{align}
\boldsymbol{C}_{\mu_l} = & \E\left\{\boldsymbol{\mu}_l^{(+)}(t){\boldsymbol{\mu}_l^{(+)}}^T(t)\right\} = \E\Bigg\{\left[\boldsymbol{\varUpsilon}\boldsymbol{e}_l(t)+\boldsymbol{h}_l\boldsymbol{J}\left[ \boldsymbol{H}\boldsymbol{A}\boldsymbol{r}(t)+\boldsymbol{e} (t) \right]\right]\nonumber \\
\times& \left[\boldsymbol{e}^T_l(t)\boldsymbol{\varUpsilon}^T+\left[\boldsymbol{r}^T(t)\boldsymbol{A}^T \boldsymbol{H}^T+\boldsymbol{e}^T (t) \right]\boldsymbol{J}^T\boldsymbol{h}^T_l\right]\Bigg\}\nonumber\\
=& \E\Bigg\{\boldsymbol{\varUpsilon}\boldsymbol{e}_l(t)\boldsymbol{e}_l^T(t)\boldsymbol{\varUpsilon}^T+\boldsymbol{\varUpsilon}\underbrace{\boldsymbol{e}_l(t)\left[\boldsymbol{r}^T(t)\boldsymbol{A}^T \boldsymbol{H}^T+\boldsymbol{e}^T (t) \right]}_{0}+\boldsymbol{h}_l\boldsymbol{J}\underbrace{\left[ \boldsymbol{H}\boldsymbol{A}\boldsymbol{r}(t)+\boldsymbol{e} (t) \right]\boldsymbol{e}^T_l(t)}_{0}\boldsymbol{\varUpsilon}^T\nonumber\\
+ & \boldsymbol{h}_l\boldsymbol{J}\left[ \boldsymbol{H}\boldsymbol{A}\boldsymbol{r}(t)\boldsymbol{r}^T(t)\boldsymbol{A}^T \boldsymbol{H}^T+\boldsymbol{e} (t)\boldsymbol{r}^T(t)\boldsymbol{A}^T \boldsymbol{H}^T+  \boldsymbol{H}\boldsymbol{A}\boldsymbol{r}(t)\boldsymbol{e}^T(t)+\boldsymbol{e} (t) \boldsymbol{e}^T(t) \right]\boldsymbol{J}^T\boldsymbol{h}^T_l  \Bigg\}\nonumber\\
= & \boldsymbol{\varUpsilon}\boldsymbol{R}_l\boldsymbol{\varUpsilon}^T+\boldsymbol{h}_l\boldsymbol{J}\left[ \boldsymbol{H}\boldsymbol{A}\boldsymbol{C}_r\boldsymbol{A}^T \boldsymbol{H}^T -\boldsymbol{R}\tilde{\boldsymbol{K}}^T\boldsymbol{Q}^T\boldsymbol{A}^T\boldsymbol{H}^T - \boldsymbol{H}\boldsymbol{A}\boldsymbol{Q}\tilde{\boldsymbol{K}}\boldsymbol{R}+\boldsymbol{R} \right]\boldsymbol{J}^T\boldsymbol{h}^T_l .
\end{align}
\end{appendix}\vspace{-5mm}
	\def\baselinestretch{1}
	\bibliographystyle{IEEEtran}
	\bibliography{references}

	
	


\end{document}